\RequirePackage{etex}
\documentclass[a4paper,USenglish,cleveref,thm-restate]{lipics-v2021}

\usepackage{cite}
\usepackage{amsmath,amssymb,mathtools}
\usepackage{tikz}

\usetikzlibrary{arrows.meta,positioning}
\usepackage[linesnumbered,ruled,vlined]{algorithm2e}
\SetCommentSty{mycommfont}
\usepackage{nicefrac}   
\SetAlgoSkip{smallskip}
\SetInd{0.5em}{0.8em}
\mathtoolsset{showonlyrefs=true}
\newcommand{\alert}[1]{\textcolor{red}{#1}}
 
\usepackage{colortbl}
\newcolumntype{d}{>{\columncolor{gray!30}}c}
\usepackage{booktabs}
\usepackage{float}
\usepackage{multirow}

\newcommand{\bA}{\mathbf{A}}
\newcommand{\cA}{\mathcal{A}}
\newcommand{\bB}{\mathbf{B}}

\newcommand{\Umax}{\mathrm{Umax}}
\newcommand{\Emax}{\mathrm{Emax}}
\newcommand{\fEmax}{\overline{\mathrm{Emax}}}
\renewcommand{\mid}{:}
\DeclareMathOperator*{\argmax}{arg\,max}

\title{Simultaneously Efficient Allocation of Indivisible Items Across Multiple Dimensions}
\author{Yasushi Kawase}{Chuo University, Tokyo, Japan}{}{https://orcid.org/0000-0001-5626-779X}{JST ERATO Grant Number JPMJER2301, JSPS KAKENHI Grant Number JP25K00137, Value Exchange Engineering, a joint research project between Mercari R4D Lab and RIISE (Research Institute for an Inclusive Society through Engineering).}
\author{Bodhayan Roy}{Indian Institute of Technology Kharagpur, Kharagpur, India}{}{https://orcid.org/0009-0005-6476-3060}{ANRF MATRICS Grant Number MTR/2021/000474}
\author{Mohammad Azharuddin Sanpui}{Indian Institute of Technology Kharagpur, Kharagpur, India}{}{https://orcid.org/0000-0001-5030-9645}{JST Sakura Science Exchange Program, JST LOTUS Programme Grant Number JPMJ25152681}
\authorrunning{Y. Kawase, B. Roy, and M. A. Sanpui}
\Copyright{Yasushi Kawase, Bodhayan Roy, and Mohammad Azharuddin Sanpui}
\relatedversion{}
\EventEditors{John Q. Open and Joan R. Access}
\EventNoEds{2}
\EventLongTitle{42nd Conference on Very Important Topics (CVIT 2016)}
\EventShortTitle{CVIT 2016}
\EventAcronym{CVIT}
\EventYear{2016}
\EventDate{December 24--27, 2016}
\EventLocation{Little Whinging, United Kingdom}
\EventLogo{}
\SeriesVolume{42}
\ArticleNo{23}

\nolinenumbers

\begin{document}
\maketitle

\begin{abstract}
Many allocation problems are intrinsically multidimensional, since an item may contribute differently to several criteria, and optimizing a single aggregate objective can hide severe losses in other dimensions. We study how much efficiency can be guaranteed simultaneously when indivisible items have multiple attributes. To this end, we introduce the \emph{multidimensional efficient allocation} (MDEA) model, where each agent has an additive valuation in each dimension, and investigate simultaneous efficiency under utilitarian social welfare (USW) and egalitarian social welfare (ESW). Our results reveal a sharp worst-case frontier. For exact efficiency, maximizing the number of dimensions attaining the USW optimum admits a $c/\ell$-approximation for every fixed constant $c$, and this dependence on the number $\ell$ of dimensions is essentially unavoidable; for ESW, even deciding whether two dimensions can be optimized simultaneously is NP-hard with binary valuations. For approximate simultaneous efficiency in every dimension, we identify a tight threshold of order $1/\ell$, showing that such guarantees always exist for both USW and ESW, while any asymptotically better dependence on $\ell$ is impossible, even for binary valuations. Finally, we introduce three natural multidimensional Pareto notions and characterize both their relationships and their computational complexity.
\end{abstract}

\newpage

\section{Introduction}\label{sec:introduction}
Resource allocation is a fundamental problem in economics and computer science~\cite{Amanatidis2022FairDO,reiter1962allocating,maskin1987fair,huesch2012one}. Classical models of indivisible-item allocation typically optimize a single welfare objective, such as utilitarian or egalitarian social welfare~\cite{aziz2022algorithmic,lipton2004approximately}. In many realistic settings, however, the quality of an allocation cannot be captured by a single scalar objective, because items must be evaluated along multiple dimensions.

Such multi-criteria considerations arise in a wide range of applications. In cloud resource allocation, virtual machines or server slots may need to be assigned while accounting simultaneously for CPU capacity, memory, bandwidth, and energy efficiency. In personnel allocation, workers may be assigned to projects or service units, each of which needs a balanced mix of skills such as technical expertise, experience, communication ability, and domain knowledge. In such settings, optimizing one criterion in isolation can substantially degrade performance in others, while aggregating all criteria into a single score may obscure important trade-offs or stakeholder priorities. Since each dimension represents a distinct requirement, ignoring any one of them may lead to unacceptable outcomes. This inherent tension makes simultaneous efficiency fundamentally challenging.

These challenges motivate the need to understand allocation mechanisms that perform well across multiple dimensions simultaneously. Our goal is to determine \emph{to what extent efficiency can be achieved simultaneously across several dimensions}.

This perspective is also reflected in recent work on fair allocation with multiple objectives. One line augments fair allocation by allowing the allocator, in addition to agents, to have preferences over outcomes~\cite{Bu2024fair,Barman2025FairDW,flammini2025fair}; another studies fair allocation with multidimensional preferences~\cite{Kawase2025SimultaneouslyFA}. These works show the importance of going beyond a single objective, but they focus primarily on fairness and are largely limited to one or two dimensions. In contrast, we study simultaneous efficiency in a model with an arbitrary number of dimensions, with the goal of characterizing the resulting dimension-dependent trade-offs, approximation guarantees, and impossibility phenomena.

To capture this setting, we introduce the \emph{Multidimensional Efficient Allocation (MDEA)} model. In MDEA, each item is described by $\ell$ attributes (or dimensions), and each agent's valuations are additive within each dimension. By abstracting away from application-specific structure, MDEA provides a clean framework for studying the dimension-dependent trade-offs that arise when efficiency must be pursued simultaneously across many objectives. We study simultaneous efficiency under two fundamental welfare criteria: \emph{utilitarian social welfare (USW)}, the sum of agents' utilities, and \emph{egalitarian social welfare (ESW)}, the minimum utility among agents. Ideally, we would like an allocation that maximizes USW or ESW in every dimension at once; we call such allocations \emph{simultaneous Umax (sUmax)} and \emph{simultaneous Emax (sEmax)}, respectively. In one dimension, Umax is optimized greedily by assigning each item to an agent who values it most, whereas Emax is the classical max-min allocation problem and is already NP-hard for two agents with ternary valuations~\cite{fitzsimmons2024hardness}.

This exact notion is already too demanding even in tiny instances. The following two examples show this for the two welfare notions we study. The first shows that maximizing total welfare in every dimension at once can be impossible, while the second shows that the same difficulty persists when the objective is to maximize the minimum welfare across agents in every dimension.

\begin{example}[Impossibility for sUmax]\label{ex:umax1-imp}
Consider two agents, two dimensions, and $2c+1$ items. Each item gives value 1 to agent 1 in dimension 1 and value 1 to agent 2 in dimension 2, with value 0 otherwise. The optimal USW is $2c+1$ in each dimension, but any allocation assigns each item to a single agent, so at least one dimension gets value at most $c$. Hence, no sUmax allocation exists for any $c\ge 0$.
\end{example}

\begin{example}[Impossibility for sEmax]\label{ex:emax1-imp}
Consider two agents, two dimensions, and $4c+2$ items. $2c+1$ items favor agent 1 in dimension 1 and agent 2 in dimension 2, and the other $2c+1$ swap these roles, each with value 1 for the favored agent and 0 otherwise. The optimal ESW is $2c+1$ in each dimension, but any allocation leaves some agent with value at most $c$ in some dimension. Hence, no sEmax allocation exists for any $c \ge 0$.
\end{example}

Together, these examples show that the obstruction is not specific to one welfare criterion: exact simultaneous efficiency can fail for both USW and ESW, although the source of failure is interpreted differently in the two cases. Rather than indicating a weakness of the model, this points to a more basic phenomenon: simultaneous efficiency across many dimensions is intrinsically limited in a general multidimensional setting. This makes exact simultaneous efficiency too strong a target in general and motivates several natural relaxations.

One possibility is to maximize the number of dimensions in which the optimum is attained, asking how many objectives can still be satisfied exactly at the same time. A second possibility is to seek the best possible simultaneous guarantee in every dimension. Since a purely multiplicative relaxation is already ruled out by \Cref{ex:umax1-imp,ex:emax1-imp}, we instead study guarantees that combine multiplicative approximation with an additive loss caused by a bounded number of high-impact goods, namely \emph{$\alpha$-sUmax up to $c$ goods ($\alpha$-sUmax$\,c$)} and \emph{$\alpha$-sEmax up to $c$ goods ($\alpha$-sEmax$\,c$)}. This loss is charged only to goods whose contribution is unavailable to the relevant requirement, rather than to the optimum of a smaller instance. For Umax this means goods not assigned to a dimension-wise maximizing agent, and for Emax it means goods outside the relevant agent's bundle. These notions are related in spirit to $\alpha$-EF1 in fair division~\cite{amanatidis2023round,Barman2025FairDW}, but here the goal is to characterize the optimal worst-case dependence on the number of dimensions. A third possibility is to study multidimensional Pareto optimality, which provides weaker but always meaningful baseline notions of efficiency. In particular, we consider three ways of comparing allocations: by checking whether every agent-dimension pair weakly improves (\emph{PO-agent}), whether the vector of USW values across dimensions weakly improves (\emph{PO-USW}), or whether the vector of ESW values across dimensions weakly improves (\emph{PO-ESW}). Together, these perspectives let us quantify the frontier of what simultaneous efficiency can and cannot guarantee across multiple dimensions.

Our results reveal a fundamental limitation: achieving efficiency across multiple dimensions requires significant trade-offs. In particular, guarantees degrade with the number of dimensions, and this degradation is unavoidable, highlighting a qualitative gap from the single-dimensional setting.

\subsection{Our Results}\label{sec:our-results}

This paper presents the first systematic study of simultaneously efficient allocation when items are characterized by multiple attributes. We establish structural, algorithmic, and hardness results for exact, approximate, and Pareto-based notions of simultaneous efficiency across multiple dimensions. Throughout, we treat the set of dimensions as given exogenously; in practical applications, choosing or learning a useful low-dimensional representation is an important complementary modeling question, and our guarantees apply to the selected representation. Let $n$ be the number of agents, $m$ be the number of items, and $\ell$ be the number of dimensions.

For \emph{exact} simultaneous efficiency, we study how many dimensions can be optimized exactly under Umax and Emax. On the positive side, maximizing the number of Umax dimensions is polynomial-time solvable when either the number of items or the number of dimensions is constant~(\Cref{thm:Umax-constant}), and admits a $c/\ell$-approximation for every fixed constant $c$~(\Cref{thm:Umax-ell-approx}). On the negative side, even with two agents and binary valuations, the problem is NP-hard to approximate within a factor of $1/\ell^{1-\epsilon}$~(\Cref{thm:sUmax-Max-Intersect}), and maximizing the number of Emax dimensions is already NP-hard with binary valuations and two dimensions~(\Cref{thm:sEmax-hard}). This shows that the dependence on $\ell$ in the approximation guarantee is essentially unavoidable. We further show that, in the worst case, the fraction of dimensions in which Umax or Emax can be attained may be exponentially small in the numbers of agents and items~(\Cref{thm:Umax-fraction,thm:Emax-fraction}), establishing strong impossibility bounds.

For \emph{approximate} simultaneous efficiency across all dimensions, we identify a sharp threshold around $1/\ell$. An $\alpha$-sUmax$\,1$ allocation always exists for $\alpha=1/\ell$, but may fail to exist for any $\alpha>1/\ell$~(\Cref{thm:a-sUmax1-exist,thm:a-sUmax1-nonexist}). Likewise, an $\alpha$-sEmax$\,\ell$ allocation always exists for $\alpha=1/\ell$, but may fail for any $\alpha>1/\ell$~(\Cref{thm:a-sEmaxl-exist,thm:a-sEmax1-nonexist}). Thus, the order $1/\ell$ is not an artifact of our analysis, but the best possible worst-case guarantee for simultaneous efficiency in every dimension of the chosen representation. Consequently, reducing attention to a smaller set of decision-critical dimensions directly yields stronger guarantees with the corresponding smaller value of $\ell$. We also study related aggregate objectives: maximizing the sum of USWs over all dimensions is polynomial-time solvable~(\Cref{thm:total-USWs}), whereas maximizing the minimum of USWs or ESWs over all dimensions is NP-hard even to approximate~(\Cref{thm:sUmax-ge1,thm:sEmax-g1}).

For \emph{multidimensional Pareto optimality}, we define and compare three natural notions: PO-agent, PO-USW, and PO-ESW. We show that PO-USW implies PO-agent, while PO-USW and PO-ESW, as well as PO-agent and PO-ESW, are incomparable~(\Cref{prop:po-usw-esw,thm:USW-implication,thm:ESW-implication}). We also establish strong computational hardness: deciding whether a given allocation satisfies PO-agent is coNP-hard already in the single-dimensional case, and deciding PO-USW or PO-ESW is coNP-hard already in the two-dimensional case~(\Cref{thm:PO-coNP,thm:po-usw-coNPhard,thm:po-esw-coNPhard}). In addition, checking and computing PO-ESW remain hard even in the binary two-dimensional setting~(\Cref{thm:binary-po-esw-coNPhard}).

A summary of the main results, including tightness information, is deferred to \Cref{tab:results} in \Cref{app:results}.

\subsection{Related Work}\label{sec:related-work}
A line of work augments fair allocation by allowing the allocator, in addition to agents, to have preferences over outcomes~\cite{Bu2024fair,Barman2025FairDW,flammini2025fair}. These settings can be viewed as closely related two-dimensional cases of our model, with one dimension induced by the allocator (or market/social criterion) and the other by agents. Bu et al.~\cite{Bu2024fair} initiated fair allocation with allocator and agent preferences, Barman et al.~\cite{Barman2025FairDW} studied fair division with a common market valuation, and Flammini et al.~\cite{flammini2025fair} considered social-impact criteria. While these works emphasize fairness guarantees together with allocator objectives, we focus instead on simultaneous efficiency across multiple dimensions.

Another related line studies allocation among \emph{groups of agents}, primarily through fairness criteria~\cite{foley1966resource,manurangsi,SegalHalevi,Suksompong2017,Conitzer2019GroupFF,berliant1992fair,Aziz2019AlmostGE,segal2019,golz2025fair}. In these models, $\nu = \nu_1 + \cdots+ \nu_n$ agents are partitioned into $n~(\geq 2)$ groups, where group $i$ has $\nu_i~(\geq 1)$ members. When all groups have the same size, each group can be regarded as a single agent with one dimension for each member. A simple example is allocation to $n$ parent-child pairs: each pair can be viewed as one agent with two dimensions, one for the parent's valuation and one for the child's valuation, so our framework captures the problem of balancing efficiency for parents and for children across all pairs. More generally, fixed-group allocation fits naturally into our multidimensional framework, while heterogeneous group sizes can still be embedded by padding smaller groups with dummy zero-valued members. Kawase et al.~\cite{Kawase2025SimultaneouslyFA} studied fair allocation with multidimensional preferences, focusing on fairness guarantees such as EF1, PROP1, and maximin-share-type notions. In contrast, the present paper studies efficiency: simultaneous welfare maximization, tight dimension-dependent approximation thresholds, and multidimensional Pareto optimality. Hence, while the models are closely related, the objectives and the resulting algorithmic and hardness questions are largely orthogonal.

Our model is also related to multi-layered cake cutting~\cite{Hosseini2020FairDO,igarashi2021envy,Sanpui2026ProportionalAO,Cloutier2009TwoplayerEM,Nyman2017FairDW,Lebert2013EnvyfreeTM,kawase2025resource}. A key structural difference is that multi-layered cake cutting typically allocates different layers at the same position to different agents, whereas in our setting all dimensions of an item are assigned together. As a result, the underlying feasibility structure and the relevant efficiency trade-offs are different.

Overall, our framework complements multidimensional allocation models that emphasize fairness or allocator objectives by focusing directly on simultaneous efficiency.

\section{Preliminaries}\label{prilimi:model}

For a positive integer $n$, let $[n]=\{1,2,\dots,n\}$. 
Let $N=[n]$ denote the set of agents, $M=\{g_1,g_2,\dots,g_m\}$ the set of indivisible items, and $L=[\ell]$ the set of dimensions.  
Each agent $i \in N$ is equipped with a valuation function $v_i\colon M \to \mathbb{R}_+^\ell$, which assigns to every item an $\ell$-dimensional vector of nonnegative real numbers.  
We use $v_{ijk} = v_i(g_j)_k$ to denote the value that agent $i$ assigns to item $g_j \in M$ in dimension $k \in L$.
An instance of the MDEA problem is represented as $(N,M,L,(v_i)_{i\in N})$.

An allocation $\bA=(A_1,A_2,\dots,A_n)$ is a partition of the items $M$ (i.e., $\bigcup_{i\in N}A_i=M$ and $A_i\cap A_{i'}=\emptyset$ for any distinct agents $i,i'\in N$), where each subset $A_i\subseteq M$ is allocated to agent $i\in N$. The total valuation of agent $i$ for her allocated set $A_i$ is defined as 
$v_i(A_i)=\sum_{g_j\in A_i} v_i(g_j)\in\mathbb{R}_+^\ell$. Specifically, the $k$th dimensional value of agent $i$ for the allocated set $A_i$ is defined as $v_i(A_i)_k=\sum_{g_j\in A_i} v_{ijk}$.

We call an instance of the MDEA problem \emph{binary} if $v_{ijk}\in\{0,1\}$ for every agent $i\in N$, every item $g_j\in M$, and every dimension $k\in L$.
We call an instance \emph{identical} if every agent has the same valuation function, i.e., $v_1=v_2=\dots=v_n$.

For each dimension $k\in L$, define the maximum utilitarian social welfare as $\Umax_k\coloneqq\max_{\bA}\sum_{i\in N}v_i(A_i)_k$.
Similarly, define the maximum egalitarian social welfare as $\Emax_k\coloneqq\max_{\bA}\min_{i\in N}v_i(A_i)_k$.
An allocation $\bA$ is \emph{sUmax} if its USW attains $\Umax_k$ for all $k\in L$, and \emph{sEmax} if its ESW attains $\Emax_k$ for all $k\in L$.
For an allocation $\bA$ and dimension $k$, let
\[
\textstyle D_k(\bA)=\{g_j\in M\mid g_j\in A_i\text{ for some }i\in N\text{ with }v_{ijk}<\max_{i'\in N}v_{i'jk}\}.
\]
Thus, $D_k(\bA)$ is the set of items whose dimension-$k$ contribution to Umax is not fully realized by $\bA$.
Equivalently, items outside $D_k(\bA)$ already contribute their full possible value to dimension-$k$ USW, so the shortfall from $\Umax_k$ is supported only on items in $D_k(\bA)$.
For a nonnegative integer $c$, we further say that $\bA$ is \emph{$\alpha$-sUmax$\,c$} if 
\begin{align}
\textstyle \sum_{i\in N}v_i(A_i)_k\ge \alpha\cdot \Umax_k - \max_{B\subseteq D_k(\bA):\,|B|\le c}\sum_{g_j\in B}\max_{i\in N}v_{ijk}
\end{align}
for all $k\in L$, and \emph{$\alpha$-sEmax$\,c$} if 
\begin{align}
\textstyle v_i(A_i)_k\ge \alpha\cdot \Emax_k - \max_{B\subseteq M\setminus A_i:\,|B|\le c}\sum_{g_j\in B}\max_{i'\in N}v_{i'jk}
\end{align}
for all $i\in N$ and $k\in L$.
These definitions do not compare the achieved welfare with the optimum of a smaller instance.
Instead, they allow an additive loss charged to at most $c$ high-impact goods that the allocation does not use for the relevant requirement.
For Umax, these are goods in $D_k(\bA)$; for Emax, the loss for agent $i$ is charged only to goods outside $A_i$.
These definitions are evaluated \emph{dimension-wise}. For a fixed dimension $k\in L$, the additive term
$\max_{B\subseteq X:\,|B|\le c}\sum_{g_j\in B}\max_{i\in N} v_{ijk}$, with $X=D_k(\bA)$ for Umax and $X=M\setminus A_i$ for Emax, selects up to $c$ items with the largest contribution in dimension~$k$, where each item is valued by the agent who values it most in that dimension. 
Thus, these notions capture robustness against a bounded number of influential indivisible goods that remain unavailable to the corresponding objective in each dimension.
Moreover, we introduce related optimization problems:
\begin{itemize}
\item Maximum Umax dimensions: $\max_{\bA}|\{k\in L\mid \sum_i v_i(A_i)_k=\Umax_k\}|$.
\item Maximum Emax dimensions: $\max_{\bA}|\{k\in L\mid \min_i v_i(A_i)_k=\Emax_k\}|$.
\end{itemize}

An allocation $\bA$ is said to be PO-agent if there is no allocation $\bB$ such that 
(i) $v_i(B_i)_k\geq v_i(A_i)_k$ for all $i\in N$ and for all $k\in L$ and 
(ii) $v_i(B_i)_k> v_i(A_i)_k$ for some $i\in N$ and for some $k\in L$.  
An allocation $\bA$ is said to be PO-USW if there is no allocation $\bB$ such that 
(i) $\sum_{i=1}^n v_i(B_i)_k\geq \sum_{i=1}^nv_i(A_i)_k$ for all $k\in L$ and 
(ii) $\sum_{i=1}^n v_i(B_i)_k> \sum_{i=1}^n v_i(A_i)_k$ for some $k\in L$.   
An allocation $\bA$ is said to be PO-ESW if there is no such allocation $\bB$ such that 
(i) $\min_{i\in N} v_i(B_i)_k\geq \min_{i\in N} v_i(A_i)_k$ for all $k\in L$ and 
(ii) $\min_{i\in N}v_i(B_i)_k> \min_{i\in N} v_i(A_i)_k$ for some $k\in L$.   
\section{Maximizing Efficient Dimensions}
In this section, we focus on maximizing the number of dimensions in which Umax or Emax is achieved.
\subsection{Maximizing Umax Dimensions}
We analyze the problem of maximizing the number of Umax dimensions.
We begin with some basic observations.
To achieve Umax value for a specific dimension $k\in L$, each item $g_j\in M$ must be allocated to an agent $i\in N$ who values it most, that is, $v_{ijk}=\max_{i'\in N}v_{i'jk}$.
Thus, it suffices to consider the corresponding binary case $(N,M,L,(\hat{v}_i)_{i\in N})$, where $\hat{v}_{ijk}=1$ if $v_{ijk}=\max_{i'}v_{i'jk}$, and $\hat{v}_{ijk}=0$ otherwise.

Then, an allocation attains Umax for dimension $k\in L$ in the original instance if and only if it achieves Umax for $k$ in the binary instance, that is, if $\sum_{i\in N}\hat{v}_i(A_i)_k=m$.

Thus, for any subset of dimensions $L'\subseteq L$, there exists an allocation that simultaneously maximizes USW for all dimensions in $L'$ if and only if, for every $g_j\in M$, there exists an agent $i$ such that $\hat{v}_{ijk}=1$ for every $k\in L'$. Hence, it is easy to check the existence of such an allocation.
Specifically, by taking $L'=L$, we can check existence of sUmax.
\begin{theorem}\label{thm:sUmax}
The existence of an sUmax allocation can be checked in polynomial time.
Moreover, if such an allocation exists, it can be found in polynomial time.
\end{theorem}

Moreover, if the number of dimensions $\ell$ is a constant, the number of subsets $L'\subseteq L$ is $2^\ell$, which is a constant number. Consequently, we can maximize the number of Umax dimensions in polynomial time by enumerating all subsets $L'\subseteq L$ and checking whether the dimensions in $L'$ can be simultaneously maximized with respect to USW.

Separately, if the number of items $m$ is constant, then we can enumerate all the possible allocations, which is at most $O(n^m)$, a polynomial. Hence, in this case, we can maximize the number of Umax dimensions by simply enumerating all possible allocations.

Thus, we obtain the following.
\begin{observation}\label{thm:Umax-constant}
If either the number of items $m$ or the number of dimensions $\ell$ is bounded by a constant, then the maximum Umax dimensions problem can be solved in polynomial time.
Moreover, there is an FPT algorithm with respect to $\ell$ for the maximum Umax dimensions problem.
\end{observation}

Further, by enumerating all subsets $L'\subseteq L$ of size at most a constant $c$, we can find a solution for the maximum Umax dimensions problem whose objective value is the minimum of $c$ and the optimum value.
Since the optimum value for the maximum Umax dimensions problem is at most the number of dimensions $\ell$, this implies a $c/\ell$-approximation algorithm.
\begin{observation}\label{thm:Umax-ell-approx}
For any fixed constant $c$, there is a $c/\ell$-approximation algorithm for the maximum Umax dimensions problem, where $\ell$ is the number of dimensions.
\end{observation}

On the other hand, the maximum Umax dimensions problem remains computationally hard even when the number of agents $n$ is a constant. Specifically, we show a $1/\ell^{1-\epsilon}$-factor inapproximability even for two agents with binary valuations. Our reduction is from \emph{MAX-Intersect}, which is known to be as hard to approximate as Max-Clique~\cite{clifford2011maximum}.

\begin{restatable}{theorem}{THMSUMAXIN}\label{thm:sUmax-Max-Intersect}
For any $\epsilon>0$, it is NP-hard to approximate the maximum Umax dimensions problem with $\ell$ dimensions within a factor $1/\ell^{1-\epsilon}$, even when there are two agents with binary valuations.
\end{restatable}
The reduction maps each pair of sets in a MAX-Intersect instance to an item and each universe element to a dimension.
Assigning an item selects one of the two sets, and a dimension attains Umax exactly when the corresponding element belongs to all selected sets.
Thus, the number of Umax dimensions equals the intersection size, giving a gap-preserving reduction.

Together, these results (\Cref{thm:Umax-ell-approx,thm:sUmax-Max-Intersect}) show that while a simple $c/\ell$-approximation is achievable, improving the $1/\ell$ dependence is unlikely, highlighting a fundamental limitation imposed by the number of dimensions.

Finally, we analyze the worst-case fraction of dimensions that can achieve Umax.
We show that the fraction is exponentially small with respect to $n$ and $m$.
\begin{restatable}{theorem}{THMUMAXFRAC}\label{thm:Umax-fraction}
Fix the number of agents $n$, the number of items $m$, and the number of dimensions $\ell$.
For any MDEA instance, the maximum number of Umax dimensions $\max_{\bA}|\{k\in L\mid \sum_{i\in N}v_i(A_i)_k=\Umax_k\}|$ is at least $\lceil \ell/n^m\rceil$.
Moreover, there exists a binary MDEA instance for which $\max_{\bA}|\{k\in L\mid \sum_{i\in N}v_i(A_i)_k=\Umax_k\}|\le\lceil \ell/n^m\rceil$. 
\end{restatable}
The lower bound follows by averaging over the $n^m$ allocations. Since each dimension has at least one Umax-achieving allocation, some allocation achieves Umax in at least $\lceil \ell/n^m \rceil$ dimensions. The matching upper bound is obtained by constructing an instance where each dimension admits a unique Umax allocation and distributing the dimensions evenly across allocations.

The $\lceil\ell/n^m\rceil$ bound is tight, showing that only a very small fraction of dimensions can be optimized simultaneously, an inherent limitation of the multidimensional setting.

\subsection{Maximizing Emax Dimensions}
We now analyze the problem of maximizing the number of Emax dimensions.
Without loss of generality, we may assume in this subsection that the number of items is at least the number of agents. Otherwise, the Emax value is zero for every dimension, and any allocation achieves sEmax.

Similar to the Umax case, when the number of items $m$ is constant, we can maximize the number of Emax dimensions by enumerating all possible allocations.
Thus, we obtain the following result.
\begin{observation}\label{thm:Emax-constant}
If the number of items $m$ is bounded by a constant, then the maximum Emax dimensions problem can be solved in polynomial time.
\end{observation}

However, unlike Umax, computing an Emax allocation is already difficult in the one-dimensional setting: with two agents, the problem is polynomial-time solvable for binary valuations, but becomes NP-hard for ternary valuations~\cite{fitzsimmons2024hardness}. For the exact simultaneous problem, the one-dimensional case is trivial, since an sEmax allocation always exists and the maximum number of Emax dimensions is one. In contrast, in two dimensions both of these decision problems become NP-hard even with binary valuations. The proof proceeds by a reduction from 3-dimensional matching (3DM).
\begin{restatable}{theorem}{THMSEAMHARD}\label{thm:sEmax-hard}
The problems of checking the existence of an sEmax allocation and of determining the maximum number of Emax dimensions are NP-hard, even when there are two dimensions and the valuations are binary.
\end{restatable}
At a high level, the reduction represents each triple by an agent and creates items corresponding to the elements of the second and third parts of the 3DM instance, together with auxiliary items for triples that are not selected.
The construction has Emax value one in each dimension, and an allocation attains both values exactly when the agents that receive the two element-items corresponding to their triples form a perfect 3DM matching.

This highlights a sharp jump in complexity: moving from one to two dimensions fundamentally changes the problem from trivial to computationally intractable.

For the non-binary case, we can prove NP-hardness of checking the existence of an sEmax allocation, and hence of determining the maximum number of Emax dimensions, when there are only \emph{two agents} and only two dimensions. 
\begin{restatable}{theorem}{TMSEMAXHARD}\label{thm:sEmax-hard2}
The problems of checking the existence of an sEmax allocation and of determining the maximum number of Emax dimensions are NP-hard even when there are only two agents and only two dimensions. 
\end{restatable}
The proof is by a reduction from \textsc{Partition}. The two dimensions encode the two target sums, and two auxiliary items force any sEmax allocation to correspond to a balanced partition.

Finally, we analyze the worst-case fraction of dimensions that can achieve Emax, as a function of the number of dimensions $\ell$. As in the case of Umax, we show that the fraction is exponentially small with respect to $n$ and $m$.
Specifically, let $T(m,n) = \sum_{i=0}^n(-1)^{n-i}\binom{n}{i}i^m$ denote the number of ways to allocate $m$ items to $n$ agents, where each agent receives at least one item. 
Then, the fraction is $\lceil \ell/T(m,n)\rceil$.

\begin{restatable}{theorem}{THMEMAXFRAC}\label{thm:Emax-fraction}
Fix the number of agents $n$, the number of items $m$, and the number of dimensions $\ell$, with $m\ge n$.
For any MDEA instance, the maximum number of Emax dimensions $\max_{\bA}|\{k\in L\mid \min_{i\in N}v_i(A_i)_k=\Emax_k\}|$ is at least $\lceil \ell/T(m,n)\rceil$.
Moreover, there exists an MDEA instance for which $\max_{\bA}|\{k\in L\mid \min_{i\in N}v_i(A_i)_k=\Emax_k\}|=\lceil \ell/T(m,n)\rceil$. 
\end{restatable}
The lower bound again follows by averaging. Among the allocations in which every agent receives at least one item, there are $T(m,n)$ possibilities, and hence some allocation achieves Emax in at least $\lceil \ell / T(m,n)\rceil$ dimensions. For the upper bound, we distribute dimensions evenly across these allocations and construct valuations so that each dimension attains Emax under only one allocation, yielding the bound.

This result provides a tight bound on the number of Emax dimensions. As with Umax, only a vanishing fraction of dimensions can be optimized simultaneously, showing that this limitation is inherent even for egalitarian objectives.
Unlike the Umax upper-bound construction, the matching construction for Emax uses non-binary valuations.

\section{Simultaneously Maximizing Efficiency}
In this section, we consider the extent to which the efficiency of each dimension must be relaxed in order to achieve simultaneous efficiency across all dimensions.

\subsection{Approximate sUmax}\label{subsec:approx-sUmax}

We show that an $\alpha$-sUmax$\,1$ allocation always exists when $\alpha=1/\ell$, whereas it may not exist when $\alpha>1/\ell$.

To establish existence, we use a round-robin procedure in which each dimension is treated as a virtual agent.
In the turn of dimension $k$, the procedure selects a remaining item maximizing $u_k(g)=\max_{i\in N}v_i(g)_k$ and assigns it to an agent attaining this maximum.
This process cycles through the dimensions until all items are allocated; the pseudocode is given as Algorithm~\ref{alg:roundrobin-sUmax1} in \Cref{app:proofs}.
This is the same round-robin principle used in the EF1 existence proof of Caragiannis et al.~\cite{caragiannis2019unreasonable}.

\begin{restatable}{theorem}{THMSUMAXEXIST}\label{thm:a-sUmax1-exist}
A $(1/\ell)$-sUmax$\,1$ allocation always exists, and such an allocation can be found in polynomial time.
\end{restatable}
The key point is that, for each dimension $k$, the items selected in the turns of $k$ dominate the later unallocated items in value $u_k(g)=\max_{i\in N}v_i(g)_k$.
Hence, the only possible loss from $\Umax_k$ comes from items selected before the first turn of $k$; after ignoring items already assigned to a dimension-$k$ maximizing agent, this loss can be charged to one item in $D_k(\bA)$.

Conversely, for every positive integer $\ell$, the approximation ratio $1/\ell$ cannot be improved even for the binary case.
\begin{restatable}{theorem}{THMSUMAXNON}\label{thm:a-sUmax1-nonexist}
Fix a positive integer $\ell$.
For any positive $\alpha>1/\ell$, there exists a binary MDEA instance with $\ell$ dimensions such that no allocation satisfies $\alpha$-sUmax$\,1$.
\end{restatable}
The lower-bound instance gives positive value in dimension $k$ only to agent $k$.
Since some agent receives at most $m/\ell$ items, the corresponding dimension obtains USW at most $m/\ell$, whereas $\Umax_k=m$ and the one-item additive correction is at most $1$; choosing $m$ large rules out every $\alpha>1/\ell$.
These results establish a sharp threshold at $1/\ell$: a constant fraction per dimension is unattainable, and the linear dependence on $\ell$ is inherent.
\subsection{Approximate sEmax}\label{subsec:approx-sEmax}

We show that an $\alpha$-sEmax$\,\ell$ allocation always exists when $\alpha=1/\ell$, whereas for $\alpha>1/\ell$ such an allocation may not exist.

A direct round-robin procedure, which cycles through all agent-dimension pairs and lets each agent pick a favorite remaining item in the corresponding dimension, gives the following weaker guarantee. Since both the approximation ratio and the additive term depend on $n$, this result is mainly a useful baseline; the pseudocode and proof are deferred to \Cref{app:proofs}.

\begin{restatable}{theorem}{THMASEAMEXIST}\label{thm:a-sEmaxnl-exist}
A $\frac{1}{n\cdot\ell}$-sEmax$\,(n\ell-1)$ allocation always exists, and such an allocation can be found in polynomial time.
\end{restatable}

To establish existence with an approximation ratio of $1/\ell$, we instead apply the iterative rounding technique of Gölz and Yaghoubizade~\cite[Theorem 4.2]{golz2025fair}, which guarantees the existence of a PROP$\ell$ allocation for fair division of indivisible items among groups of agents, each of size at most $\ell$. In our context, a group of $\ell$ agents can be regarded as a setting with $\ell$ dimensions. The corresponding pseudocode is given as Algorithm~\ref{alg:sEmax-iterative} in \Cref{app:proofs}.

\begin{restatable}{theorem}{THMSEAXEXIST}\label{thm:a-sEmaxl-exist}
A $\frac{1}{\ell}$-sEmax$\,\ell$ allocation always exists, and such an allocation can be found in polynomial time.
\end{restatable}
The proof starts from optimal fractional Emax allocations, one for each dimension.
Averaging these fractional solutions guarantees each agent a $1/\ell$ fraction of the fractional optimum in every dimension.
Iterative rounding then converts the averaged fractional allocation into an integral one, losing value only on at most $\ell$ items outside each relevant bundle.

Conversely, for every positive integer $\ell$, the approximation ratio $1/\ell$ cannot be improved even for the binary case.
\begin{restatable}{theorem}{THMSEMAXNONEXIST}\label{thm:a-sEmax1-nonexist}
Fix a positive integer $\ell$.
For any positive $\alpha>1/\ell$ and positive integer $c$, there exists a binary MDEA instance with $\ell$ agents and $\ell$ dimensions such that no allocation satisfies $\alpha$-sEmax$\,c$.
\end{restatable}
The construction is a cyclic binary instance in which, for every dimension, each item is valuable to exactly one agent, and the responsible agent shifts with the dimension.
Although $\Emax_k=m/n$ for every dimension, any allocation gives some agent at most $m/n$ items; averaging over the $\ell$ dimensions for this agent yields a dimension with value at most $m/(n\ell)$, so no $\alpha>1/\ell$ guarantee survives after removing any fixed number $c$ of items.

These results establish a tight $1/\ell$ threshold for Emax: even with additive relaxations, the dependence on the number of dimensions is unavoidable.

\subsection{Other Notions}
Finally, we discuss other possible measures of efficiency across multiple dimensions.

The problem of maximizing the sum of USWs over all dimensions can be solved greedily by assigning each item $g$ to an agent $i$ who maximizes $\sum_{k\in L}v_{ijk}$.
\begin{theorem}\label{thm:total-USWs}
The problem of maximizing the sum of USWs over all dimensions is solvable in polynomial time.
\end{theorem}

On the other hand, the problem of maximizing the minimum of USWs over all dimensions (i.e., $\min_{k\in L}\sum_{i\in N} v_{i}(A_i)_k$) is NP-hard to approximate, even in the binary case, as shown below.

\begin{restatable}{theorem}{THMSUMAXGE}\label{thm:sUmax-ge1}
Checking the existence of an allocation $\bA$ such that $\min_{k\in L}\sum_{i\in N}v_i(A_i)_k\ge 1$ is NP-complete, even when the valuations are binary. 
Moreover, maximizing the same objective is NP-hard to approximate, even when the valuations are binary. 
\end{restatable}
The reduction is from \emph{Hitting Set}: agents represent universe elements, dimensions represent sets, and the $h$ items select at most $h$ agents.
An allocation achieves value at least one in every dimension exactly when the selected agents hit all sets.
The same gap gives the inapproximability statement.

If we do not restrict the problem to the binary case, the problem remains NP-hard even when $n=\ell=2$.
\begin{restatable}{theorem}{THMSUMAXLL}\label{thm:sUmax-22}
Even when the numbers of agents and dimensions are two, finding an allocation $\bA$ that maximizes $\min_{k\in L}\sum_{i\in N}v_i(A_i)_k$ is NP-hard.
\end{restatable}
The proof is again by a reduction from \textsc{Partition}, using the two dimensions to encode the two sides of the partition.

Nevertheless, if each valuation is an integer bounded by a constant (i.e., in $\{0,1,\dots,c\}$ for some constant $c$) and the number of dimensions $\ell$ is fixed, the problem of maximizing the minimum of USWs over all dimensions is solvable in polynomial time by a dynamic programming approach.
Specifically, one constructs a table of possible $\ell$-dimensional USW vectors for each prefix set of items $\{g_1,g_2,\dots,g_j\}$ for every $j\in[m]$.
Since the number of possible USW vectors is at most $(cj+1)^\ell$, this is polynomial for fixed $\ell$, with the exponent depending on $\ell$. 
\begin{theorem}
  There is a polynomial-time algorithm for the problem of maximizing the minimum of USWs over all dimensions if each valuation is an integer bounded by a constant and the number of dimensions is fixed.
\end{theorem}

For ESW, consider two problems: maximizing the sum of ESWs over all dimensions (i.e., $\sum_{k\in L}\min_{i\in N} v_{i}(A_i)_k$) and maximizing the minimum of ESWs over all dimensions (i.e., $\min_{k\in L}\min_{i\in N} v_{i}(A_i)_k$). 
However, since computing an Emax allocation is NP-hard already in the one-dimensional two-agent case with ternary valuations~\cite{fitzsimmons2024hardness}, both of these problems are NP-hard even when the number of dimensions is one.
Moreover, \Cref{thm:sEmax-hard} implies that the problem of maximizing the minimum ESW across dimensions is NP-hard even to approximate. Therefore, egalitarian objectives are inherently hard, and this difficulty is further amplified in multidimensional settings.
\begin{theorem}\label{thm:sEmax-g1}
Checking the existence of an allocation $\bA$ such that $v_i(A_i)_k\ge 1$ for every $i\in N$ and $k\in L$ is NP-complete, even when there are two dimensions and the valuations are binary.
Moreover, maximizing $\min_{k\in L}\min_{i\in N}v_i(A_i)_k$ is NP-hard to approximate, even when there are two dimensions and the valuations are binary.
\end{theorem}
This follows from the same 3DM reduction as \Cref{thm:sEmax-hard}: in that construction, the Emax value is one in both dimensions, so achieving value at least one for every agent and dimension is exactly the existence of an sEmax allocation. The resulting gap between value one and value zero gives the inapproximability statement.

\section{Pareto Optimality}
In this section, we examine the concepts of PO-agent, PO-USW, and PO-ESW.  

\subsection{Relationship between PO Notions}

We observe relationships between PO-agent, PO-USW, and PO-ESW.

First, PO-USW and PO-ESW are fundamentally distinct concepts, and they are not necessarily compatible.  
\begin{restatable}{proposition}{PROPOUESW}\label{prop:po-usw-esw}
There exists a single-dimensional MDEA instance in which no allocation simultaneously satisfies both PO-USW and PO-ESW.
\end{restatable}

Comparing PO-agent and PO-USW, we show that PO-USW is stronger than PO-agent.
\begin{restatable}{theorem}{THMUSWIM}\label{thm:USW-implication}
In every MDEA instance, every PO-USW allocation is also PO-agent.
However, some MDEA instances admit PO-agent allocations that are not PO-USW.
\end{restatable}

It may seem that the same relationship holds between PO-agent and PO-ESW as between PO-agent and PO-USW.
However, the two concepts are incomparable.
Nevertheless, we can show that there always exists an allocation that simultaneously satisfies both PO-agent and PO-ESW.
\begin{restatable}{theorem}{THMESWIM}\label{thm:ESW-implication}
PO-agent and PO-ESW are incomparable.
However, for every MDEA instance, there exists an allocation that satisfies both PO-agent and PO-ESW simultaneously.
\end{restatable}
The existence part follows by taking a PO-agent allocation that is maximal with respect to PO-ESW among all PO-agent allocations; any ESW improvement can be followed by agent-Pareto improvements, contradicting this maximality.

\subsection{Checking Pareto Optimality}
The multidimensional efficiency measures considered above (e.g., maximum Umax dimensions, $\alpha$-sUmax$\,1$) are compatible with PO-agent because these properties are preserved under Pareto improvements.
We therefore ask whether the algorithms studied above can guarantee PO-agent, PO-USW, or PO-ESW.

First, note that the algorithm for \Cref{thm:total-USWs} gives a PO-agent and PO-USW allocation, since it maximizes the sum of utilities with respect to agents and USWs.
\begin{observation}\label{obs:compute-PO}
A PO-agent and PO-USW allocation can be computed in polynomial time.
\end{observation}
However, most algorithms do not necessarily yield PO allocations.
Consider an instance of two agents, two dimensions, and two items $g_1$, $g_2$. 
The valuations are $v_1(g_1)=(5,1)$, $v_1(g_2)=(2,1)$ and $v_2(g_1)=(4,3)$, $v_2(g_2)=(0,2)$. The round-robin algorithm described in \Cref{thm:a-sUmax1-exist} gives $\bA=(\{g_1\},\{g_2\})$. However, $\bA'=(\{g_2\},\{g_1\})$ Pareto dominates $\bA$ in the sense of USW because the USWs for $\bA$ and $\bA'$ are $(5,3)$ and $(6,4)$, respectively.
Thus, approximation guarantees should not be conflated with Pareto optimality: an allocation may satisfy the target approximation guarantee while still admitting a Pareto improvement under a stronger welfare-vector comparison.

This naturally suggests attempting to apply a Pareto improvement to a given allocation. 
Unfortunately, determining whether any Pareto improvement exists is computationally hard. In fact, deciding whether a proposed allocation is already Pareto optimal (i.e., admits no improvement) is coNP-complete, even in one dimension.
\begin{theorem}[{\cite[Theorem~$1$]{de2009complexity}}]\label{thm:PO-coNP}
Checking whether a given allocation is PO-agent is coNP-complete, even for the single-dimensional case.
\end{theorem}
Thus, even the basic task of certifying the absence of an improving allocation is intractable.
We show similar hardness results for both PO-USW and PO-ESW below.

We first establish hardness for checking PO-USW.
Note that the problem is easy in the single-dimensional case, because determining whether an allocation is PO-USW can be done by comparing it with the Umax value.
However, the problem is coNP-hard even for two dimensions.
\begin{restatable}{theorem}{THMPOUSW}\label{thm:po-usw-coNPhard}
Checking whether a given allocation is PO-USW is coNP-complete, even in the case of two agents and two dimensions.
\end{restatable}
The hardness proof reduces from \textsc{Partition}. The constructed allocation fails to be PO-USW exactly when there is an improving allocation whose two-dimensional USW vector corresponds to a balanced partition.

Next, we turn to PO-ESW. In the single-dimensional case, computing a PO-ESW allocation is equivalent to computing an Emax allocation, and is therefore NP-hard. We additionally show that checking whether a given allocation is PO-ESW is already coNP-hard with two agents and two dimensions.
\begin{restatable}{theorem}{THMPOESW}\label{thm:po-esw-coNPhard}
Computing a PO-ESW allocation is NP-hard, even in the single-dimensional two-agent case.
Additionally, checking whether a given allocation is PO-ESW is coNP-complete, even in the case of two agents and two dimensions.
\end{restatable}
The computing hardness follows from the known NP-hardness of one-dimensional Emax allocation. The verification hardness is shown by a reduction from \textsc{Partition}, where a Pareto improvement in the ESW vector exists exactly when the input can be partitioned evenly.

For binary valuations, the single-dimensional case remains tractable because PO-ESW coincides with Emax.
This tractability disappears already with two dimensions: computing a PO-ESW allocation is NP-hard, and checking whether a given allocation is PO-ESW is coNP-hard, as established by \Cref{thm:sEmax-hard}.
\begin{theorem}\label{thm:binary-po-esw-coNPhard}
Computing a PO-ESW allocation is NP-hard, even in the case of two dimensions with binary valuations.
Additionally, checking whether a given allocation is PO-ESW is coNP-complete, even in the case of two dimensions with binary valuations.
\end{theorem}
The reduction uses the same binary two-dimensional 3DM construction as \Cref{thm:sEmax-hard}.

\section{Conclusion}
In this paper, we introduced the MDEA model for allocating indivisible items with multidimensional valuations and characterized the worst-case frontier of simultaneous efficiency. We showed that exact simultaneous efficiency is severely limited, that approximate simultaneous efficiency admits sharp dimension-dependent thresholds, and that natural multidimensional Pareto notions have distinct structural and computational properties. In particular, the order $1/\ell$ for guaranteeing simultaneous efficiency in every dimension is essentially best possible in the general model.

These results should not be viewed as indicating a weakness of the model, but rather as showing that simultaneous efficiency across many dimensions is intrinsically constrained without additional structure. In this sense, MDEA serves as a baseline framework: once the worst-case frontier is identified, one can ask which extra assumptions permit stronger guarantees.

An important modeling issue concerns the choice of the dimensions themselves. We treat the set of dimensions as given exogenously, but in practice the value of $\ell$ and the interpretation of each dimension depend on how one represents the relevant criteria. Since our guarantees deteriorate with $\ell$, deciding which attributes should be kept separate, aggregated, or compressed is itself a central design question. In particular, if a decision maker can identify a smaller set of decision-critical dimensions, our guarantees apply with respect to that reduced value of $\ell$, giving stronger worst-case bounds for the selected representation. More generally, one may reduce many raw attributes to a smaller set of latent dimensions using PCA, factor analysis, or domain-driven grouping before applying our framework. Our guarantees should then be understood relative to the chosen representation.

Other directions include richer objective functions, indivisible chores, and empirical studies. \Cref{thm:sUmax} and \Cref{thm:Umax-ell-approx} extend to chores and even to mixed goods and chores, but for other notions the achievable forms of multidimensional efficiency remain open.

\bibliography{refs}

\appendix
\section{Our Results}\label{app:results}

Table~\ref{tab:results} summarizes our results.
When a bound is best possible, this is stated directly in the result description rather than encoded by a separate checkmark.

\begin{table}[H]
\centering
\caption{Summary of simultaneous efficiency results}
\label{tab:results}
\small
\setlength{\tabcolsep}{3pt}
\begin{tabular}{@{}>{\raggedright\arraybackslash}p{0.20\textwidth}>{\raggedright\arraybackslash}p{0.61\textwidth}>{\raggedright\arraybackslash}p{0.14\textwidth}@{}}
\toprule
\textbf{Setting} & \textbf{Main result} & \textbf{Ref.} \\
\midrule
Exact sUmax & Need not exist, even with two agents and two dimensions. & Ex.~\ref{ex:umax1-imp} \\
Exact sEmax & Need not exist, even with two agents and two dimensions. & Ex.~\ref{ex:emax1-imp} \\
\midrule
Max Umax dims & $c/\ell$-approximation for every fixed $c$; no $1/\ell^{1-\epsilon}$-approximation unless P${}={}$NP, so the dependence on $\ell$ is tight up to $\epsilon$. & Obs.~\ref{thm:Umax-ell-approx}, Thm.~\ref{thm:sUmax-Max-Intersect} \\
Max Emax dims & NP-hard already with binary valuations and two dimensions. & Thm.~\ref{thm:sEmax-hard} \\
\midrule
$\alpha$-sUmax$\,1$ & Always exists for $\alpha=1/\ell$, and may fail for every $\alpha>1/\ell$; this gives a tight threshold. & Thms.~\ref{thm:a-sUmax1-exist}, \ref{thm:a-sUmax1-nonexist} \\
\midrule
$\alpha$-sEmax$\,c$ & Always exists for $\alpha=1/\ell$ when $c=\ell$, while every fixed $c$ may fail for $\alpha>1/\ell$; this gives a tight threshold in $\ell$. & Thms.~\ref{thm:a-sEmaxl-exist}, \ref{thm:a-sEmax1-nonexist} \\
\midrule
Aggregate USW & Maximizing the sum is polynomial-time solvable, but maximizing the minimum is NP-hard to approximate. & Thms.~\ref{thm:total-USWs}, \ref{thm:sUmax-ge1} \\
Aggregate ESW & Maximizing the minimum ESW is NP-hard to approximate, even with binary valuations and two dimensions. & Thm.~\ref{thm:sEmax-g1} \\
\midrule
Pareto notions & PO-agent and PO-USW allocations are computable, but verifying PO-USW and PO-ESW is coNP-complete. & Obs.~\ref{obs:compute-PO}, Thms.~\ref{thm:po-usw-coNPhard}, \ref{thm:po-esw-coNPhard} \\
\bottomrule
\end{tabular}
\end{table}

\section{Omitted Proofs}\label{app:proofs}
\THMSUMAXIN*
\begin{proof}
We provide a gap-preserving reduction from the \emph{MAX-Intersect} problem. 
In the problem, we are given a universe $U=\{e_1,e_2,\dots,e_\ell\}$, 
and $m$ sets consisting of two sets $S_1,S_2,\dots,S_m$ where $S_j=\{S_{j,1},S_{j,2}\}$, the goal is to select $\pi\colon[m]\to[2]$ that maximizes $|\bigcap_{j\in[m]} S_{j,\pi(j)}|$.
By a gap-preserving reduction from the Max-Clique problem, Clifford and Popa~\cite{clifford2011maximum} proved that this problem cannot be approximated within a multiplicative factor of $1/\ell^{1-\epsilon}$, for any $\epsilon>0$, unless P${}={}$NP.

We use MAX-Intersect instances after the standard preprocessing that deletes every element $e_k$ for which $e_k\notin S_{j,1}\cup S_{j,2}$ for some $j\in[m]$.
Such an element can never belong to any feasible intersection, and hence does not affect the objective value.
We denote the size of the remaining universe by $\ell$; the hardness statement above is with respect to this reduced universe size.
From a given Max-Intersect instance, we construct an MDEA instance with two agents $N=\{1,2\}$, items $M=\{g_1,g_2,\dots,g_m\}$, and dimensions $L=[\ell]$. 
For $i\in N$, $g_j\in M$, and $k\in L$, the valuation is defined as 
\begin{align}
v_{ijk}&=\begin{cases}
1&\text{if } e_k\in S_{j,i},\\
0&\text{otherwise}.
\end{cases}
\end{align}
Note that, for each dimension $k\in L$, the Umax value is $\Umax_k=m$.

We now show that the optimum value for the Max-Intersect instance is equal to the optimum value of the maximum Umax dimensions problem for the constructed MDEA instance.
From a mapping $\pi\colon[m]\to[2]$, we construct the allocation $\bA=(A_1,A_2)$ such that $A_i=\{g_j\in M\mid \pi(j)=i\}$ for $i=1,2$.
Then, the USW value $\sum_{i\in N}v_i(A_i)_k$ attains $\Umax_k~(=m)$ if and only if $e_k\in S_{j,\pi(j)}$ for all $j\in[m]$. This means that the number of Umax dimensions for $\bA$ is equal to $|\bigcap_{j\in[m]}S_{j,\pi(j)}|$.
Conversely, given an allocation $\bA'=(A'_1,A'_2)$, we can construct a mapping $\pi'\colon [m]\to[2]$ as $\pi'(j)=i$ if $g_j\in A'_i$ for $i=1,2$. The number of Umax dimensions for $\bA'$ is then $|\bigcap_{j\in[m]}S_{j,\pi'(j)}|$.
Therefore, this is a gap-preserving reduction.
\end{proof}
\THMUMAXFRAC*
\begin{proof}

Let $\cA$ denote the set of all possible allocations of the $m$ items to the $n$ agents. Since each item can be assigned to any of the $n$ agents, the total number of distinct allocations is $|\cA| = n^m$.
Define $\cA_k\subseteq \cA$ to be the set of allocations $\bA$ such that $\sum_{i\in N}v_i(A_i)_k=\Umax_k$ for each dimension $k\in L$.

Note that $\cA_k$ is not empty for each $k\in L$ because $\Umax_k$ is achievable by definition.
Thus, $\sum_{k=1}^\ell|\cA_k|$ is at least $\ell$.
On the other hand, for each allocation $\bA \in \cA$, let $f(\bA)=|\{k\in[\ell]\mid \bA\in\cA_k\}|$ be the number of dimensions in which $\bA$ achieves Umax. 
Then, we have
\begin{align}
\sum_{\bA \in \cA} f(\bA) = \sum_{k=1}^\ell |\cA_k| 
\geq \ell.
\end{align}
Thus, by the pigeonhole principle, we obtain
\[
\textstyle
\max\limits_{\bA\in\cA}\Big|\big\{k\in L\mid \sum\limits_{i\in N}v_i(A_i)_k=\Umax_k\big\}\Big|
=\max\limits_{\bA\in\cA}f(\bA)
\geq \left\lceil \frac{\ell}{n^m} \right\rceil.
\]
This proves the lower bound.

Next, we provide the upper bound.
Choose a mapping $\rho\colon [\ell]\to\cA$ such that $|\{k\in[\ell]\mid \rho(k)=\bA\}|\le \lceil\ell/n^m\rceil$ for every $\bA\in\cA$. 
Such a mapping exists by distributing the $\ell$ dimensions as evenly as possible among the $|\cA|=n^m$ allocations.
For each agent $i\in[n]$, item index $j\in[m]$, and dimension $k\in[\ell]$,
we define the valuation as
\begin{align}
v_{ijk}=\begin{cases}
1 & \text{if agent $i$ receives $g_j$ in allocation $\rho(k)$},\\
0 & \text{otherwise}.
\end{cases}
\end{align}
This instance is binary by construction.

For this instance, the Umax value in dimension $k\in[\ell]$ is $m$, and it can only be achieved by using the allocation $\rho(k)$.
Hence, for each allocation $\bA\in\cA$, the number of dimensions in which $\bA$ achieves Umax is $|\{k\in[\ell]\mid \rho(k)=\bA\}|\le\lceil \ell/n^m\rceil$.
Thus, $\max_{\bA}|\{k\in L\mid \sum_{i\in N}v_i(A_i)_k=\Umax_k\}|\le\lceil \ell/n^m\rceil$ for this instance. 
\end{proof}
\THMSEAMHARD*
\begin{proof}
It suffices to prove hardness for the existence of an sEmax allocation, since the constructed instances have two dimensions and hence attaining Emax in both dimensions is equivalent to maximizing the number of Emax dimensions to two.
We present a polynomial-time reduction from the 3-dimensional matching (3DM) problem, which is known to be NP-hard~\cite{garey1979computers}.
In the 3DM problem, we are given three sets of elements, $X=\{x_1,\dots,x_n\}$, $Y=\{y_1,\dots,y_n\}$, and $Z=\{z_1,\dots,z_n\}$. We are also given a set of hyperedges $T=\{t_1,\dots,t_m\}$ where each $t\in T$ is an ordered triple in $X\times Y\times Z$. The goal of the problem is to determine if there exists a subset $T'$ of $T$ such that each element from $X$, $Y$, and $Z$ appears exactly once in $T'$.
In other words, our task is to find a perfect matching that covers all the elements in $X$, $Y$, and $Z$ without any repetitions.
Without loss of generality, we assume that there exists a subset $T^{(Y)} \subseteq T$ such that each element from $X$ and $Y$ appears exactly once in $T^{(Y)}$, since otherwise the instance is a no-instance, and this can be easily verified in polynomial time.  
Similarly, we assume that there exists a subset $T^{(Z)} \subseteq T$ such that each element from $X$ and $Z$ appears exactly once in $T^{(Z)}$.

For each $j\in [n]$, let $s_j=|\{t_i\in T\mid x_j\in t_i\}|$ be the number of hyperedges $t_i$ that contains $x_j$.
Note that $\sum_{j=1}^n s_j=m$.
From a given 3DM instance, we construct an MDEA instance with agents $N=[m]$, $M=\{g_1,g_2,\dots,g_{n+m}\}$, and $L=\{1,2\}$.
For $i\in N$, $g_j\in M$, $k\in L$, with $t_i=(x_a,y_b,z_c)$,
the valuation $v_{ijk}$ equals $1$ if $(j,k)=(b,1)$, $(n+c,2)$, or $2n+\sum_{p=1}^{a-1}(s_p-1)< j\le 2n+\sum_{p=1}^{a}(s_p-1)$, and $0$ otherwise.
An example of this reduction is illustrated in \Cref{tab:sEmax-hard}.

\begin{table}[htbp]
\centering
\caption{The valuations $v_{ijk}$ for the reduced instance in \Cref{thm:sEmax-hard} for the instance: $n=3$, $m=5$, $t_1=(x_1,y_2,z_2)$, $t_2=(x_1,y_2,z_3)$, $t_3=(x_2,y_3,z_3)$, $t_4=(x_3,y_1,z_1)$, $t_5=(x_3,y_3,z_3)$. The red allocation corresponds to the solution $\{t_1,t_3,t_4\}$.}\label{tab:sEmax-hard}
\begin{tabular}{cc||c|c|c||c|c|c||c|c}
\toprule
agent & dim.
& $g_1$
& $g_2$
& $g_3$ 
& $g_4$ 
& $g_5$ 
& $g_6$ 
& $g_7$ 
& $g_8$ \\\hline
\multirow{2}{*}{1} & 1 & 0 & \alert{1} & 0 & 0 & \alert{0} & 0 & 1 & 0\\\cline{2-10}
                   & 2 & 0 & \alert{0} & 0 & 0 & \alert{1} & 0 & 1 & 0\\\hline
\multirow{2}{*}{2} & 1 & 0 & 1 & 0 & 0 & 0 & 0 & \alert{1} & 0\\\cline{2-10}
                   & 2 & 0 & 0 & 0 & 0 & 0 & 1 & \alert{1} & 0\\\hline
\multirow{2}{*}{3} & 1 & 0 & 0 & \alert{1} & 0 & 0 & \alert{0} & 0 & 0\\\cline{2-10}
                   & 2 & 0 & 0 & \alert{0} & 0 & 0 & \alert{1} & 0 & 0\\\hline
\multirow{2}{*}{4} & 1 & \alert{1} & 0 & 0 & \alert{0} & 0 & 0 & 0 & 1\\\cline{2-10}
                   & 2 & \alert{0} & 0 & 0 & \alert{1} & 0 & 0 & 0 & 1\\\hline
\multirow{2}{*}{5} & 1 & 0 & 0 & 1 & 0 & 0 & 0 & 0 & \alert{1}\\\cline{2-10}
                   & 2 & 0 & 0 & 0 & 0 & 0 & 1 & 0 & \alert{1}\\
\bottomrule
\end{tabular}
\end{table}

We remark that, for each dimension, the Emax value is $1$, and each agent must receive exactly one valuable item.
Indeed, for each $i \in N$ with $t_i = (x_a, y_b, z_c) \in T$,  
allocating $g_b$ to $i$ if $t_i \in T^{(Y)}$, and allocating one item from  
$\big\{g_{2n+\sum_{p=1}^{a-1}(s_p-1)+1}, \dots, g_{2n+\sum_{p=1}^{a}(s_p-1)}\big\}$ if $t_i \notin T^{(Y)}$,  
guarantees a utility of $1$ for agent $i$ in the first dimension.
Moreover, the Emax value is at most $1$ for the first dimension because there are $m$ items that are valuable to at least one agent in the first dimension (i.e., $\{g_1,\dots,g_n,g_{2n+1},\dots,g_{n+m}\}$), and each agent can receive at most one valuable item. 
An analogous argument applies to the second dimension.

We now prove that the MDEA instance admits a sEmax allocation if and only if the 3DM instance is a yes-instance.
Suppose that an allocation $\bA$ achieves sEmax, i.e., $v_i(A_i)_k=1$ for all $i\in N$ and $k\in L$.
Then, for each $i\in N$ with $t_i = (x_a, y_b, z_c) \in T$, 
we have $A_i=\{g_b,g_{n+c}\}$ or $A_i=\{g_j\}$ with $2n+\sum_{p=1}^{a-1}(s_p-1)< j\le 2n+\sum_{p=1}^{a}(s_p-1)$.
Thus, $T'\coloneqq\{i\in N\mid |A_i|=2\}$ forms a 3-dimensional matching, and the 3DM instance is a yes-instance.
Conversely, suppose that the 3DM instance is a yes-instance, and let $T'$ be a 3-dimensional matching.
Then, for each $i \in N$ with $t_i = (x_a, y_b, z_c) \in T$,  
allocating $\{g_b,g_{n+c}\}$ to $i$ if $t_i \in T'$, and allocating one item from  
$\{g_{2n+\sum_{p=1}^{a-1}(s_p-1)+1}, \dots, g_{2n+\sum_{p=1}^{a}(s_p-1)}\}$ if $t_i \notin T'$,  
guarantees a utility of $1$ for agent $i$ in both dimensions.
Thus, this allocation achieves the ESW value of 1 for both dimensions.
Therefore, the problems of checking the existence of an sEmax allocation and of determining the maximum number of Emax dimensions are NP-hard, even when there are two dimensions and the valuations are binary.
\end{proof}

\TMSEMAXHARD*
\begin{proof}
We present a polynomial-time reduction from the PARTITION problem, which is known to be NP-complete~\cite{garey1979computers}.
In the PARTITION problem, we are given $m$ positive integers $a_1,a_2,\cdots,a_m$ such that $\sum_{j=1}^{m}a_j=2T$. The goal is to determine whether there exists a partition $(S_1,S_2)$ of the set $[m]$ such that $\sum_{j\in S_1}a_j=\sum_{j\in S_2}a_j=T$. 

From a given instance of the PARTITION problem, we construct an MDEA instance with agents $N=\{1,2\}$, items $M=\{g_1,g_2,\dots,g_m,g_{m+1},g_{m+2}\}$, and dimensions $L=\{1,2\}$.
For $i\in N$, $g_j\in M$, $k\in L$, the valuation is defined as
\begin{align}
v_{ijk}=\begin{cases}
a_j & \text{if $i=k$ and $j\le m$},\\
T   & \text{if $(i,j,k)=(1,m+1,2)$ or $(2,m+2,1)$},\\
0   & \text{otherwise}.
\end{cases}
\end{align}
The valuations are illustrated in \Cref{tab:sEmax-hard2}.

\begin{table}[htbp]
\centering
\caption{The valuations $v_{ijk}$ for the reduced instance in \Cref{thm:sEmax-hard2}. Each column block corresponds to an item, the blue numbers denote the agent index, and the two rows represent the valuations in dimensions $1$ and $2$, respectively.}\label{tab:sEmax-hard2}
\tabcolsep = 1.2mm
\begin{tabular}{c||c|c|c|c|c|c|c|c|c|c|c|c|c|c|}
\toprule
& \multicolumn{2}{c|}{$g_1$} 
& \multicolumn{2}{c|}{$g_2$} 
& $\cdots$
& \multicolumn{2}{c|}{$g_j$} 
& $\cdots$
& \multicolumn{2}{c|}{$g_m$} 
& \multicolumn{2}{c|}{$g_{m+1}$}
& \multicolumn{2}{c|}{$g_{m+2}$} \\\hline
& $\textcolor{blue}{1}$ & $\textcolor{blue}{2}$ &  $\textcolor{blue}{1}$ & $\textcolor{blue}{2}$ & $\cdots$ & $\textcolor{blue}{1}$ & $\textcolor{blue}{2}$ & $\cdots$ &  $\textcolor{blue}{1}$ & $\textcolor{blue}{2}$ & $\textcolor{blue}{1}$ & $\textcolor{blue}{2}$ & $\textcolor{blue}{1}$ & $\textcolor{blue}{2}$
\\\hline
$1$ & $a_1$ & $0$ & $a_2$ & $0$ & $\cdots$ & $a_j$ & $0$ & $\cdots$ & $a_m$ & $0$ & $0$ & $0$ & $0$ & $T$  \\
$2$ & $0$ & $a_1$ & $0$ & $a_2$ & $\cdots$ & $0$ & $a_j$ & $\cdots$ & $0$ & $a_m$ & $T$& $0$ & $0$ & $0$ \\
\bottomrule
\end{tabular}
\end{table}

Without loss of generality, we may assume that $g_{m+1}$ is allocated to agent $1$ and $g_{m+2}$ is allocated to agent $2$. The Emax values are $\Emax_1=\Emax_2=T$.
If the PARTITION instance is a yes-instance, then there exists a partition $(S_1,S_2)$ of $[m]$ such that $\sum_{j\in S_1}a_j=\sum_{j\in S_2}a_j=T$.
Allocating $\{g_j\mid j\in S_1\}\cup\{g_{m+1}\}$ to agent $1$ and $\{g_j\mid j\in S_2\}\cup\{g_{m+2}\}$ to agent $2$ achieves Emax in both dimensions.

Conversely, suppose that the MDEA instance admits a sEmax allocation $\bA$.
To achieve $\Emax_1=T$, we must have $v_1(A_1)_1=\sum_{g_j\in A_1\cap\{g_1,\dots,g_m\}}a_j\ge T$, since $g_{m+1}$ contributes nothing in dimension $1$ and $g_{m+2}$ is allocated to agent $2$.
Similarly, achieving $\Emax_2=T$ requires $v_2(A_2)_2=\sum_{g_j\in A_2\cap\{g_1,\dots,g_m\}}a_j\ge T$.
Because $\sum_{j=1}^m a_j=2T$ and the items $\{g_1,\dots,g_m\}$ are partitioned between the two agents, both sums must in fact be exactly $T$.
Hence, the sets $S_1=A_1\cap\{g_1,\dots,g_m\}$ and $S_2=A_2\cap\{g_1,\dots,g_m\}$ form a solution to the PARTITION instance.

Since there are only two dimensions, determining whether the maximum number of Emax dimensions is two is equivalent to checking the existence of an sEmax allocation.
Therefore, the problems of checking the existence of an sEmax allocation and of determining the maximum number of Emax dimensions are NP-hard.
\end{proof}
\THMEMAXFRAC*
\begin{proof}
Let $\cA'$ denote the set of all possible allocations of the $m$ items to the $n$ agents, where each agent receives at least one item. Note that $|\cA'|=T(m,n)$.
Define $\cA'_k\subseteq \cA'$ to be the set of allocations $\bA\in\cA'$ such that $\min_{i\in N}v_i(A_i)_k=\Emax_k$ for each dimension $k\in L$.

Note that $\cA'_k$ is not empty for each $k\in L$ because $\Emax_k$ is achievable while ensuring that no agent receives zero items.
Thus, $\sum_{k=1}^\ell|\cA'_k|$ is at least $\ell$.
On the other hand, for each allocation $\bA \in \cA'$, let $f'(\bA)=|\{k\in[\ell]\mid \bA\in\cA'_k\}|$ be the number of dimensions in which $\bA$ achieves Emax. 
Then, we have
\begin{align}
\sum_{\bA \in \cA'} f'(\bA) = \sum_{k=1}^\ell |\cA'_k| 
\geq \ell.
\end{align}
Thus, by the pigeonhole principle, we obtain
\begin{align}
\max_{\bA\in\cA'}\left|\left\{k\in L\mid \min_{i\in N}v_i(A_i)_k=\Emax_k\right\}\right|
=\max_{\bA\in\cA'}f'(\bA) 
\geq \left\lceil \ell/T(m,n) \right\rceil.
\end{align}
This proves the lower bound.

Next, we provide the upper bound.
Choose a mapping $\rho'\colon [\ell]\to\cA'$ such that $|\{k\in[\ell]\mid \rho'(k)=\bA\}|\le \lceil\ell/T(m,n)\rceil$ for every $\bA\in\cA'$. 
Such a mapping exists by distributing the $\ell$ dimensions as evenly as possible among the $|\cA'|=T(m,n)$ allocations.
For each agent $i\in[n]$, item index $j\in[m]$, and dimension $k\in[\ell]$,
let $\rho'(k)=\bA$ and define the valuation as
\begin{align}
v_{ijk}=\begin{cases}
1/|A_i| & \text{if }g_j\in A_i,\\
0       & \text{otherwise}.
\end{cases}
\end{align}
This upper-bound construction is not binary in general because the values $1/|A_i|$ may be fractional; this is the point where the construction differs from the Umax counterpart.

For this instance, the Emax value in dimension $k\in[\ell]$ is $1$, and it can only be achieved by using the allocation $\rho'(k)$.
Hence, for each allocation $\bA\in\cA'$, the number of dimensions in which $\bA$ achieves Emax is $|\{k\in[\ell]\mid \rho'(k)=\bA\}|\le\lceil \ell/T(m,n)\rceil$.
Thus, $\max_{\bA}|\{k\in L\mid \min_{i\in N}v_i(A_i)_k=\Emax_k\}|\le\lceil \ell/T(m,n)\rceil$ for this instance. 
\end{proof}

\begin{algorithm}[H]
\caption{Round-Robin $(1/\ell)$-sUmax$\,1$}
\label{alg:roundrobin-sUmax1}
\KwIn{Agents $N$, items $M$, dimensions $L=\{1,\dots,\ell\}$}
\KwOut{Allocation $\bA$}

Initialize $A_i \gets \emptyset$ for all $i \in N$\;
$M' \gets M$\;

\While{$M' \neq \emptyset$}{
    \For{$k \gets 1$ \KwTo $\ell$}{
        \If{$M' = \emptyset$}{break}
        $g^* \in \arg\max_{g \in M'} \max_{i \in N} v_i(g)_k$\;
        $i^* \in \arg\max_{i \in N} v_i(g^*)_k$\;
        $A_{i^*} \gets A_{i^*} \cup \{g^*\}$\;
        $M' \gets M' \setminus \{g^*\}$\;
    }
}
\Return{$\bA$}
\end{algorithm}

\THMSUMAXEXIST*
\begin{proof}
For each dimension $k\in L$, define the value of item $g\in M$ as $u_k(g)=\max_{i\in N}v_i(g)_k$.
We then apply a round-robin procedure as follows:
In the order $1,2,\dots,\ell$, each dimension $k$ sequentially takes turns selecting its favorite available item $g$ (i.e., an item maximizing $u_k(g)$), and $g$ is then allocated to an agent in $\argmax_{i\in N}v_{i}(g)_k$. This process is repeated in multiple rounds until all items have been allocated.
Let $\bA$ be the resulting allocation. The procedure clearly runs in polynomial time.

We now show that $\bA$ is a $(1/\ell)$-sUmax$\,1$ allocation.

Relabel the item, and let $g_j$ denote the item chosen at the $j$th step.
The items assigned in turn $k\in L$ are $g_k,g_{\ell+k},g_{2\ell+k},\dots,g_{\lfloor (m-k)/\ell\rfloor\cdot \ell+k}$. 
It follows that $\sum_{i\in N}v_i(A_i)_k\ge \sum_{p=0}^{\lfloor (m-k)/\ell\rfloor}u_k(g_{p\ell+k})$.
Note that, for each $k\in L$, $p\in\{0,1,\dots,\lfloor (m-k)/\ell\rfloor\}$, and $q\in[m]$ with $q\ge p\ell+k$, we have $u_k(g_{p\ell+k})\ge u_k(g_q)$.
Therefore,
\begin{align}
\Umax_k
&=\sum_{g\in M}\max_{i\in N}v_i(g)_k
=\sum_{g\in M}u_k(g)
=\sum_{q=1}^mu_k(g_q)\\
&\le\sum_{q=1}^{k-1}u_k(g_q) + \ell\cdot\sum_{p=0}^{\lfloor (m-k)/\ell\rfloor}u_k(g_{p\ell+k})\\
&\le \sum_{q=1}^{k-1}u_k(g_q) + \ell\cdot \sum_{i\in N}v_i(A_i)_k.
\end{align}
Let $P_k=\{q\in\{1,\dots,k-1\}\mid g_q\in D_k(\bA)\}$.
Any item among $g_1,\dots,g_{k-1}$ that is not in $D_k(\bA)$ is assigned to a dimension-$k$ maximizing agent, and hence its full $u_k$-value is already counted in $\sum_i v_i(A_i)_k$.
Therefore, the preceding bound can be sharpened to
\[
\Umax_k\le \sum_{q\in P_k}u_k(g_q)+\ell\cdot \sum_{i\in N}v_i(A_i)_k.
\]
Hence, for every $k\in L$,
\begin{align}
\sum\limits_{i\in N}v_i(A_i)_k
&\ge \frac{\Umax_k}{\ell}-\frac{\sum_{q\in P_k}u_k(g_q)}{\ell}\\
&\ge \frac{\Umax_k}{\ell}-\max_{B\subseteq D_k(\bA):\,|B|\le 1}\sum_{g_j\in B}\max_{i\in N}v_{ijk}.
\end{align}
The last inequality holds because $|P_k|\le k-1\le \ell-1$, so the average loss after division by $\ell$ is at most the largest item in $D_k(\bA)$, or zero if $D_k(\bA)=\emptyset$.
This shows that $\bA$ is a $(1/\ell)$-sUmax$\,1$ allocation.
\end{proof}

\THMSUMAXNON*
\begin{proof}
Consider an instance with $n=\ell$ agents, $m> 1/(\alpha-1/\ell)$ items, and $\ell$ dimensions.
For $i\in N$, $g_j\in M$, and $k\in L$, set the valuation $v_{ijk}$ to $1$ if $i=k$ and $0$ otherwise.

For every $k\in L$, we have $\Umax_k=m$.
Fix any allocation $\bA$.
Then, there is an agent $i^*\in N$ such that $|A_{i^*}|\le m/n=m/\ell$.
Thus,
\begin{align}
\MoveEqLeft
\sum_{i\in N}v_{i}(A_i)_{i^*}
=|A_{i^*}|\le \frac{m}{\ell}=\alpha\cdot m-\left(\alpha-\frac{1}{\ell}\right)\cdot m\\
&<\alpha\cdot m-1
\le\alpha\cdot\Umax_{i^*}-\max_{B\subseteq D_{i^*}(\bA):\,|B|\le 1}\sum_{g\in B}\max_{i\in N}v_i(g)_{i^*},
\end{align}
which implies that $\bA$ is not $\alpha$-sUmax$\,1$ for the chosen instance.
\end{proof}

\begin{algorithm}[H]
\caption{Round-Robin Allocation for $\frac{1}{n\ell}$-sEmax$\,(n\ell-1)$}
\label{alg:roundrobin-sEmax}
\KwIn{Agents $N$, items $M$, dimensions $L=\{1,\dots,\ell\}$, valuations $v_i(g)_k$}
\KwOut{Allocation $\mathbf{A}=(A_1,\dots,A_n)$}

Initialize $A_i \gets \emptyset$ for all $i \in N$\;
$M' \gets M$\;

\While{$M' \neq \emptyset$}{
    \For{$k = 1$ \KwTo $\ell$}{
        \For{$i = 1$ \KwTo $n$}{
            \If{$M' = \emptyset$}{
                \textbf{break}\;
            }
            Select $g^* \in \arg\max_{g \in M'} v_i(g)_k$\;
            Assign $g^*$ to agent $i$: $A_i \gets A_i \cup \{g^*\}$\;
            Remove $g^*$ from $M'$\;
        }
    }
}
\Return $\mathbf{A}$\;
\end{algorithm}

\THMASEAMEXIST*
\begin{proof}
We apply a round-robin procedure as follows:
Agents select items by cycling through all agents for the first dimension, then all agents for the second dimension, and so on, until the $\ell$th dimension is reached.

Formally, the order of selection is $(i,k)$ where $k$ is from $1$ to $\ell$, and for each fixed $k$, is from $1$ to $n$, that is, $(i,k)=(1,1),(2,1),\dots,(n,1),(1,2),(2,2),\dots,(n,2),\dots,(1,\ell),(2,\ell),\dots,(n,\ell)$.
In each step, agent $i$ chooses their most preferred available item $g$ with respect to $k$th dimension (i.e., an item maximizing $v_i(g)_k$). The process is repeated in multiple round until all items have been allocated.
Let $\bA$ be the resulting allocation.
The procedure clearly runs in polynomial time.

We now show that $\bA$ is a $\frac{1}{n\ell}$-sEmax$\,(n\ell-1)$ allocation.

Relabel the item, and let $g_j$ denote the item chosen at the $j$th step. The items assigned agent $i\in N$ are $g_i,g_{n+i},g_{2n+i},\dots,g_{\lfloor (m-i)/n\rfloor\cdot n+i}$. 
It follows that, for each $i\in N$ and $k\in L$, we have 
$$v_i(A_i)_k\ge \sum_{p=0}^{\lfloor (m-(k-1)n-i)/(n\ell)\rfloor}v_i(g_{pn\ell+(k-1)n+i})_k.$$
Note that, for each $i\in N$, $k\in L$, $p\in\{0,1,\dots,\lfloor (m-(k-1)n-i)/(n\ell)\rfloor\}$, and $q\in[m]$ with $q\ge pn\ell+(k-1)n+i$, we have $v_i(g_{pn\ell+(k-1)n+i})_k\ge v_i(g_q)_k$.
Therefore, for every $i\in N$ and $k\in L$,
\begin{align}
\Emax_k
&\le \sum_{g\in M}v_i(g)_k
= \sum_{q=1}^m v_i(g_q)_k\\
&=\sum_{q=1}^{(k-1)n+i-1}v_i(g_q)_k + n\ell\cdot\hspace{-4mm}\sum_{p=0}^{\lfloor (m-(k-1)n-i)/(n\ell)\rfloor}\hspace{-4mm}v_i(g_{pn\ell+(k-1)n+i})_k\\
&\le \sum_{q=1}^{(k-1)n+i-1}v_i(g_q)_k + n\ell\cdot v_i(A_i)_k.
\end{align}
Fix $i\in N$ and $k\in L$, and let
$P_{i,k}=\{g_1,\dots,g_{(k-1)n+i-1}\}\setminus A_i$.
Items in the prefix that already belong to $A_i$ are counted in $v_i(A_i)_k$, so the preceding bound can be strengthened to
\begin{align}
\Emax_k\le \sum_{g\in P_{i,k}}v_i(g)_k+n\ell\cdot v_i(A_i)_k.
\end{align}
Here $P_{i,k}\subseteq M\setminus A_i$ and $|P_{i,k}|\le n\ell-1$, and each term $v_i(g)_k$ is at most $\max_{i'\in N}v_{i'}(g)_k$.
Therefore, for every $i\in N$ and $k\in L$,
\begin{align}
v_i(A_i)_k
\ge \frac{\Emax_k}{n\ell}-\max_{B\subseteq M\setminus A_i:\,|B|\le n\ell-1}\sum_{g_j\in B}\max_{i'\in N}v_{i'jk},
\end{align}
which shows that $\bA$ is a $\frac{1}{n\ell}$-sEmax$\,(n\ell-1)$ allocation.
\end{proof}

\THMSEAXEXIST*
\begin{proof}
For each $k \in L$, we first compute the Emax value for the fractional allocation, which can be formulated as the following linear program (LP):
\begin{align}
\begin{array}{rll}
\max        & \gamma &\\
\text{s.t.} & \gamma \le \sum_{g_j\in M}v_{ijk}x_{ij} & (i\in N),\\[5pt]
            &\sum_{i\in N}x_{ij}=1                    & (g_j\in M),\\[5pt]
            &0\le x_{ij}\le 1                         & ((i,j)\in N\times M).
\end{array}
\end{align}
Let $\fEmax_k$ denote the optimal value and $x^{(k)}$ the corresponding optimal solution of the LP.
Every integral allocation is feasible for this LP, and hence $\fEmax_k\ge\Emax_k$.
Here $\fEmax_k$ is the fractional optimum, whereas $\Emax_k$ denotes the integral optimum used in the definition of sEmax.
Furthermore, $x^*=\sum_{k'\in L}x^{(k')}/\ell$ is a fractional allocation that guarantees each agent at least a $1/\ell$ fraction of $\fEmax_k$ for every $k\in L$. We now construct an integral allocation with a matching additive-loss guarantee.

We now use the iterative rounding framework of Gölz and Yaghoubizade~\cite[Theorem 4.2]{golz2025fair}, and briefly recall its idea. Starting from a basic feasible solution of a polytope, it repeatedly fixes decisions forced by the remaining fractional structure: either an item is almost fully assigned to one agent and is fixed, or some agent constraint has support size at most $\ell$ and those items are fixed together. The polytope is updated after each step; feasibility is preserved and the loss is bounded by the sum of the largest values of at most $\ell$ items, yielding an integral allocation. We apply the framework with our target guarantee by replacing the proportionality condition $v_i(M)_k/n$ (written as $u_{gi}(M)/n$ in their notation) with $\fEmax_k/\ell$, which is feasible for $x^*$.
Then the algorithm outputs an allocation $\bA$ such that
\begin{align}
v_i(A_i)_k\ge \frac{\fEmax_k}{\ell}-\max_{B\subseteq M\setminus A_i:\,|B|\le \ell}\sum_{g_j\in B}\max_{i'\in N}v_{i'jk}
\end{align}
holds for each $i\in N$ and $k\in L$.
Since $\fEmax_k\ge\Emax_k$, for every $i\in N$ and $k\in L$ we have
\begin{align}
v_i(A_i)_k
&\ge \frac{\Emax_k}{\ell}-\max\limits_{B\subseteq M\setminus A_i:\,|B|\le \ell}\sum\limits_{g_j\in B}\max\limits_{i'\in N}v_{i'jk}.
\end{align}
Thus, $\bA$ is a $\frac{1}{\ell}$-sEmax$\,\ell$ allocation.
\end{proof}

\begin{algorithm}[H]
\caption{Iterative Rounding for $(1/\ell)$-sEmax$\,\ell$}
\label{alg:sEmax-iterative}
\KwIn{Agents $N$, items $M$, dimensions $L=\{1,\dots,\ell\}$, valuations $v_i(g)_k$}
\KwOut{Allocation $\bA=(A_1,\dots,A_n)$}

Initialize $A_i \gets \emptyset$ for all $i \in N$\;

\For{each dimension $k \in L$}{
    Solve LP:
    \[
    \max \gamma \quad
    \text{s.t. } \gamma \le \sum_{g \in M} v_i(g)_k x_{ig} \ \forall i,\;
    \sum_{i} x_{ig} = 1,\;
    0 \le x_{ig} \le 1
    \]
    Let $x^{(k)}$ be the optimal solution\;
}

Compute averaged fractional allocation:
\[
x^*_{ig} \gets \frac{1}{\ell}\sum_{k \in L} x^{(k)}_{ig}
\]

\While{there exists fractional $x^*$}{
    \If{there exists item $g$ with $x^*_{ig}=1$ for some $i$}{
        Assign $g$ to $i$: $A_i \gets A_i \cup \{g\}$\;
        Remove $g$ from the instance and update $x^*$\;
    }
    \Else{
        Find agent $i$ with at most $\ell$ fractional items, whose existence follows from the iterative-rounding lemma~\cite[Theorem~4.2]{golz2025fair}\;
        Assign all such items to $i$\;
        Update instance and $x^*$\;
    }
}

\Return $\bA$\;
\end{algorithm}

\THMSEMAXNONEXIST*
\begin{proof}
Consider an instance with $n=\ell$ agents, $m> cn/(\alpha-1/\ell)$ items, and $\ell$ dimensions.
Suppose that $m$ is a multiple of $n$.
For $i\in N$, $g_j\in M$, and $k\in L$, set the valuation $v_{ijk}$ to $1$ if $j\equiv i+k-1\pmod{n}$ and $0$ otherwise.

For every $k\in L$, the Emax value is $\Emax_k=m/n$ by allocating each item $g_j$ to the agent $i\equiv j-k+1\pmod{n}$.
Fix any allocation $\bA$.
Then, there is an agent $i^*\in N$ such that $|A_{i^*}|\le m/n$.
Moreover, there is a dimension $k^*\in L$ such that $v_{i^*}(A_{i^*})_{k^*}\le |A_{i^*}|/\ell\le m/(n\ell)$.
Indeed, for fixed $i^*$, each item contributes value $1$ to agent $i^*$ in exactly one dimension, so $\sum_{k\in L}v_{i^*}(A_{i^*})_k=|A_{i^*}|$ and the claim follows by averaging.
Thus,
\begin{align}
v_{i^*}(A_{i^*})_{k^*}
&\le \frac{m}{n\ell}=\alpha \frac{m}{n}-\Big(\alpha-\frac{1}{\ell}\Big)\cdot \frac{m}{n}
<\alpha\cdot \frac{m}{n}-c\\
&\le\alpha\Emax_{k^*}-\max_{B\subseteq M\setminus A_{i^*}:\,|B|\le c}\sum_{g\in B}\max_{i\in N}v_i(g)_{k^*},
\end{align}
which implies that $\bA$ is not $\alpha$-sEmax$\,c$ for the chosen instance.
\end{proof}

\THMSUMAXGE*
\begin{proof}
The decision problem is in NP since a proposed allocation can be checked directly.

We present a polynomial-time reduction from the \emph{Hitting set problem}. 
In the Hitting set problem, we are given a universe $U=\{a_1,a_2,\dots,a_n\}$, subsets $S_1,S_2,\dots,S_m\subseteq U$, and an integer $h$. The goal is to check existence of a subset $U'\subseteq U$ of size $h$ such that $U'\cap S_p\ne \emptyset$ for all $p\in [m]$. It is known that this problem is NP-complete~\cite{garey1979computers}.

From a given instance of the Hitting set problem,
we construct an MDEA instance with agents $N=[n]$, items $M=\{g_1,g_2,\dots,g_h\}$, and dimensions $L=[m]$.
For $i\in N$, $g_j\in M$, and $p\in L$, the valuation $v_{ijp}$ is $1$ if $a_i\in S_p$ and $0$ otherwise.
We remark that the valuation $v_{ijp}$ does not depend on the item $g_j$. 

Intuitively, any agent who receives at least one item in the MDEA instance corresponds to an element selected as $U'$ in the Hitting set instance.
We now show that the USW $\sum_{i\in N}v_i(A_i)_p$ is at least $1$ for every $p\in L$ if and only if the Hitting set instance is a yes-instance.

Suppose that the Hitting set instance is a yes-instance, i.e., there exists a subset $U'\subseteq U$ with $|U'|=h$ such that $U'\cap S_p\ne\emptyset$ for all $p\in [m]$. 
Let $U'=\{a_{q_1},a_{q_2},\dots,a_{q_h}\}$, and consider an allocation $\bA$ such that $A_{q_t}=\{g_t\}$ for each $t\in[h]$ and $A_i=\emptyset$ for each $a_i\in U\setminus U'$.
Then, for each dimension $p\in L$, the USW is $\sum_{i\in N}v_i(A_i)_p=|\{a_i\in U'\mid a_i\in S_p\}|\ge 1$.

Conversely, suppose that there exists an allocation $\bA$ such that $\sum_{i\in N} v_i(A_i)_p\ge 1$ for all $p\in L$. 
Let $U'=\{a_i\in U\mid A_i\ne\emptyset\}$ (if the cardinality of $U'$ is smaller than $h$, we arbitrarily add elements to $U'$ so that its size becomes $h$).
Then, $U'\cap S_p\ne\emptyset$ by $\sum_{i\in N}v_i(A_i)_p\ge 1$. 

Therefore, the problem of deciding whether the maximum value of $\min_{p\in L}\sum_{i\in N} v_i(A_i)_p$ is at least one or zero is NP-complete. This establishes the claimed hardness results.
   
\end{proof}

\THMSUMAXLL*
\begin{proof}
We present a polynomial-time reduction from the PARTITION problem, which is known to be NP-complete~\cite{garey1979computers}.
In the PARTITION problem, we are given $m$ positive integers $a_1,a_2,\cdots,a_m$ such that $\sum_{j=1}^{m}a_j=2T$. The goal is to determine whether there exists a partition $(S_1,S_2)$ of the set $[m]$ such that $\sum_{j\in S_1}a_j=\sum_{j\in S_2}a_j=T$. 

From a given instance of the PARTITION problem, we construct an MDEA instance with agents $N=\{1,2\}$, items $M=\{g_1,g_2,\dots,g_m\}$, and dimensions $L=\{1,2\}$.
For each agent $i\in N$, the valuation for item $g_j\in M$ and dimension $k\in L$ is defined as
\begin{align}
v_{ijk}=\begin{cases}
a_j & \text{if }i=k,\\
0   & \text{if }i\ne k.
\end{cases}
\end{align}
For an allocation $\bA$, we have 
$v_1(A_1)_1\ge T$ if and only if $\sum_{g_j\in A_1}a_j\ge T$ and 
$v_2(A_2)_2\ge T$ if and only if $\sum_{g_j\in A_2}a_j\ge T$.
Thus, the MDEA instance admits an allocation $\bA$ with $\min_{k\in L}\sum_{i\in N}v_i(A_i)_k\ge T$ if and only if the PARTITION instance is a yes-instance.
Therefore, finding an allocation $\bA$ that maximizes $\min_{k\in L}\sum_{i\in N}v_i(A_i)_k$ is NP-hard.

\end{proof}

\PROPOUESW*
\begin{proof}
    Consider an MDEA instance with agents $N=\{1,2\}$, items $M=\{g_1, g_2\}$, and dimensions $L=\{1\}$. 
    Suppose that the valuations are given by $v_1(g_1)_1=v_1(g_2)_1=2$, $v_2(g_1)_1=0$, and $v_2(g_2)_1=1$. 
    Then, allocation $(\{g_1,g_2\},\emptyset)$ is the unique PO-USW allocation, with USW value $4$.
    In contrast, $(\{g_1\},\{g_2\})$ is the unique PO-ESW allocation, with ESW value $1$.
    Therefore, PO-USW and PO-ESW are not compatible in this instance.   
\end{proof}

\THMUSWIM*
\begin{proof}
We begin by proving the contrapositive of the first statement.
Let $\bA$ be an allocation that is not PO-agent in an MDEA instance $(N,M,L,(v_i)_{i\in N})$.
Then, there exists an allocation $\bA'$ that Pareto dominates $\bA$ with respect to agents, i.e., 
$v_i(A'_i)_k\ge v_i(A_i)_k$ for every $i\in N$, $k\in L$, and
$v_{i^*}(A'_{i^*})_{k^*}> v_{i^*}(A_{i^*})_{k^*}$ for some $i^*\in N$, $k^*\in L$.
This implies that $\bA'$ also Pareto dominates $\bA$ with respect to USW because 
$\sum_{i\in N}v_i(A'_i)_k\ge \sum_{i\in N}v_i(A_i)_k$ for every $k\in L$, and
$\sum_{i\in N}v_{i}(A'_{i})_{k^*}> \sum_{i\in N}v_{i}(A_{i})_{k^*}$ for $k^*\in L$.

Next, we prove the second statement.
Consider an MDEA instance with agents $N=\{1,2\}$, items $M=\{g\}$, and dimensions $L=\{1\}$. 
Suppose that the valuations are given by $v_1(g)_1=2$ and $v_2(g)_1=1$. 
Then, allocation $(\{g\},\emptyset)$ is the unique PO-USW allocation, with USW value $2$.
However, $(\emptyset,\{g\})$ is also a PO-agent allocation.
Therefore, PO-agent does not necessarily imply PO-USW.
\end{proof}
\THMESWIM*
\begin{proof}
For the first statement, consider an MDEA instance with agents $N=\{1,2\}$, items $M=\{g_1,g_2,g_3\}$, and dimensions $L=\{1\}$. 
Suppose that the valuations are given by 
$(v_1(g_1)_1,v_1(g_2)_1,v_1(g_3)_1)=(3,2,1)$ and 
$(v_2(g_1)_1,v_2(g_2)_1,v_2(g_3)_1)=(1,3,0)$.
Then, allocation $\bA^{(1)}=(\{g_1,g_2,g_3\},\emptyset)$ is PO-agent, but not PO-ESW.
In contrast, allocation $\bA^{(2)}=(\{g_1\},\{g_2,g_3\})$ is PO-ESW, but not PO-agent.
The allocation $\bA^{(3)}=(\{g_1,g_3\},\{g_2\})$ satisfies both PO-agent and PO-ESW.
Note that $\bA^{(3)}$ Pareto dominates $\bA^{(1)}$ with respect to ESW and $\bA^{(2)}$ with respect to agents.

For the second statement, consider the finite set of allocations that satisfy PO-agent, and choose one, say $\bA$, that is maximal with respect to PO-ESW among this set.
Such an allocation exists because every finite partially ordered set has a maximal element.
We claim that $\bA$ is also PO-ESW among all allocations.
Suppose otherwise that some allocation $\bB$ Pareto dominates $\bA$ with respect to ESW.
Starting from $\bB$, repeatedly apply Pareto improvements with respect to agents until reaching an allocation $\bA^\star$ that is PO-agent.
Then, for every $k\in L$, each agent-Pareto improvement weakly increases every $v_i(A_i)_k$, and hence weakly increases $\min_{i\in N}v_i(A_i)_k$.
Therefore, $\bA^\star$ Pareto dominates $\bB$ with respect to ESW, and thus also Pareto dominates $\bA$ with respect to ESW.
This contradicts the choice of $\bA$ as a PO-ESW-maximal allocation among all PO-agent allocations.
Hence, $\bA$ satisfies both PO-agent and PO-ESW.
\end{proof}

\THMPOUSW*
\begin{proof}
The problem is in coNP since a Pareto improvement is a polynomial-size certificate that a given allocation is not PO-USW.

We provide a reduction from the PARTITION problem.
Suppose that we are given $m$ positive integers $a_1,a_2,\cdots,a_m$ such that $\sum_{j=1}^{m}a_j=2T$. 

From a given instance of the PARTITION problem, we construct an MDEA instance with agents $N=\{1,2\}$, items $M=\{g_1,g_2,\dots,g_m,g_{m+1}\}$, and dimensions $L=\{1,2\}$.
For $i\in N$, $g_j\in M$, $k\in L$, the valuation is defined as
\begin{align}
v_{ijk}=\begin{cases}
a_j   & \text{if $i=k$ and $j\le m$},\\
T     & \text{if $(i,j,k)=(1,m+1,1)$},\\
T-1/2 & \text{if $(i,j,k)=(2,m+1,2)$},\\
0     & \text{otherwise}.
\end{cases}
\end{align}
The valuations are illustrated in \Cref{tab:po-usw-coNPhard}.
Note that the USWs for an allocation $\bA$ are $v_1(A_1)_1$ and $v_2(A_2)_2$ for the first and the second dimension, respectively.
Consider the problem of checking whether $\bA^*\coloneqq(\{g_1,\dots,g_m\},\{g_{m+1}\})$ is a PO-USW allocation.

\begin{table}[htbp]
\centering
\caption{The valuations $v_{ijk}$ for the reduced instance in \Cref{thm:po-usw-coNPhard}.}\label{tab:po-usw-coNPhard}
\tabcolsep = 1.2mm
\begin{tabular}{c||c|c|c|c|c|c|c|c|c|c}
\toprule
& \multicolumn{2}{c|}{$g_1$} 
& $\cdots$
& \multicolumn{2}{c|}{$g_j$} 
& $\cdots$
& \multicolumn{2}{c|}{$g_m$} 
& \multicolumn{2}{c|}{$g_{m+1}$}\\\hline
& $\textcolor{blue}{1}$ & $\textcolor{blue}{2}$ & $\cdots$ & $\textcolor{blue}{1}$ & $\textcolor{blue}{2}$ & $\cdots$ &  $\textcolor{blue}{1}$ & $\textcolor{blue}{2}$ & $\textcolor{blue}{1}$ & $\textcolor{blue}{2}$ \\\hline
$1$ & $a_1$ & $0$   & $\cdots$ & $a_j$ & $0$   & $\cdots$ & $a_m$ & $0$   & $T$ & $0$    \\
$2$ & $0$   & $a_1$ & $\cdots$ & $0$   & $a_j$ & $\cdots$ & $0$   & $a_m$ & $0$  & $T-1/2$ \\
\bottomrule
\end{tabular}
\end{table}

Suppose that $\bA'$ Pareto improves $\bA^*$ with respect to USW.
Then, we have $v_1(A'_1)_1\ge v_1(A^*_1)_1=2T$, $v_2(A'_2)_2\ge v_2(A^*_2)_2=T-1/2$, and $v_1(A'_1)_1+v_2(A'_2)_2> v_1(A^*_1)_1+v_2(A^*_2)_2=3T-1/2$.
This is possible only when $v_1(A'_1)_1=2T$ and $v_2(A'_2)_2=T$.
Consequently, for $I=\{j\in[m]\mid g_j\in A'_2\}$, we have $\sum_{j\in I}a_j=T$.
Therefore, such an allocation $\bA'$ exists if and only if the PARTITION instance is a YES-instance.
Hence, checking whether a given allocation is PO-USW is coNP-complete, even in the case of two agents and two dimensions.
\end{proof}

\THMPOESW*
\begin{proof}
The first statement follows from the NP-hardness of computing an Emax allocation in the single-dimensional two-agent case. We prove only the second statement.

The problem is in coNP since a Pareto improvement is a polynomial-size certificate that a given allocation is not PO-ESW.

We provide a reduction from the PARTITION problem.
Suppose that we are given $m$ positive integers $a_1,a_2,\cdots,a_m$ such that $\sum_{j=1}^{m}a_j=2T$. 

From a given instance of the PARTITION problem, we construct an MDEA instance with agents $N=\{1,2\}$, items $M=\{g_1,g_2,\dots,g_m,g_{m+1}\}$, and dimensions $L=\{1,2\}$.
For $i\in N$, $g_j\in M$, $k\in L$, the valuation is defined as
\begin{align}
v_{ijk}=\begin{cases}
a_j    & \text{if $i=1$ and $j\le m$},\\
T      & \text{if $i=1$ and $j=m+1$},\\
2a_j   & \text{if $i=2$ and $j\le m$},\\
2T-1/2 & \text{if $i=2$ and $j=m+1$},\\
0      & \text{otherwise}.
\end{cases}
\end{align}
The valuations are illustrated in \Cref{tab:po-esw-coNPhard}.
Consider the problem of checking whether $\bA^*\coloneqq(\{g_1,\dots,g_m\},\{g_{m+1}\})$ is a PO-ESW allocation.
Note that the ESW value for $\bA^*$ is $\min\{v_1(A^*_1)_1,v_2(A^*_2)_1\}=\min\{2T,2T-1/2\}=2T-1/2$.

\begin{table}[htbp]
\centering
\caption{The valuations $v_{ijk}$ for the reduced instance in \Cref{thm:po-esw-coNPhard}.}\label{tab:po-esw-coNPhard}
\tabcolsep = 1.2mm
\begin{tabular}{c||c|c|c|c|c|c|c|c|c|c}
\toprule
& \multicolumn{2}{c|}{$g_1$} 
& $\cdots$
& \multicolumn{2}{c|}{$g_j$} 
& $\cdots$
& \multicolumn{2}{c|}{$g_m$} 
& \multicolumn{2}{c|}{$g_{m+1}$}\\\hline
& $\textcolor{blue}{1}$ & $\textcolor{blue}{2}$ & $\cdots$ & $\textcolor{blue}{1}$ & $\textcolor{blue}{2}$ & $\cdots$ &  $\textcolor{blue}{1}$ & $\textcolor{blue}{2}$ & $\textcolor{blue}{1}$ & $\textcolor{blue}{2}$ \\\hline
$1$ & $a_1$ & $a_1$   & $\cdots$ & $a_j$ & $a_j$   & $\cdots$ & $a_m$ & $a_m$   & $T$ & $T$    \\\hline
$2$ & $2a_1$ & $2a_1$   & $\cdots$ & $2a_j$ & $2a_j$   & $\cdots$ & $2a_m$ & $2a_m$   & $2T-1/2$ & $2T-1/2$    \\
\bottomrule
\end{tabular}
\end{table}

Suppose that $\bA'$ has a strict better ESW value than $\bA^*$.
Then, we have $v_1(A'_1)_1\ge 2T$ and $v_2(A'_2)_1\ge 2T$.
This is possible only when $v_1(A'_1)_1=2T$ and $v_2(A'_2)_1=2T$.
Consequently, for $I=\{j\in[m]\mid g_j\in A'_2\}$, we have $\sum_{j\in I}a_j=T$.
Therefore, such an allocation $\bA'$ exists if and only if the PARTITION instance is a YES-instance.
Hence, checking whether a given allocation is PO-ESW is coNP-complete, even in the case of two agents and two dimensions.
\end{proof}

\end{document}